\documentclass[12pt, twoside]{article}
\usepackage{latexsym,amssymb,amsmath,amsthm,amsfonts,enumerate,verbatim,
xspace,exscale}
\usepackage{graphicx}
\usepackage{color,amsbsy}
\usepackage{times}
 
\usepackage{hyperref}
\definecolor{darkblue}{rgb}{0,0,.5}
\hypersetup{pdftex=true, colorlinks=true, breaklinks=true,%
linkcolor=darkblue, menucolor=darkblue, pagecolor=darkblue,%
urlcolor=darkblue,citecolor=darkblue}

\input xy 
\xyoption{all} 
\CompileMatrices
\UseComputerModernTips

\def\titlerunning#1{\gdef\titrun{#1}}
\makeatletter
\def\author#1{\gdef\autrun{\def\and{\unskip, }#1}\gdef\@author{#1}}
\def\address#1{{\def\and{\\\hspace*{18pt}}\renewcommand{\thefootnote}{}%
\footnote {#1}}%
\markboth{\autrun}{\titrun}}
\makeatother
\def\email#1{e-mail: #1}
\def\subjclass#1{{\renewcommand{\thefootnote}{}%
\footnote{\emph{Mathematics Subject Classification (2010):} #1}}}
\def\keywords#1{\par\medskip
\noindent\textbf{Keywords.} #1}

\theoremstyle{plain}
\newtheorem{theorem}{Theorem}[section]
\newtheorem{lemma}[theorem]{Lemma}
\newtheorem{proposition}[theorem]{Proposition}
\newtheorem{corollary}[theorem]{Corollary}

\theoremstyle{definition}
\newtheorem{definition}[theorem]{Definition}

\numberwithin{equation}{section}

\def\D{\mathcal{D}}
\def\C{\mathcal{C}}
\def\R{\mathbb{R}}

\def\N{\mathcal{N}}
\def\F{\mathcal{F}}

\def\A{\mathcal{A}}
\def\Ham{\mathcal{H}}

\providecommand{\bysame}{\makebox[3em]{\hrulefill}\thinspace}

\newcommand{\up}{\upshape}

\newcommand{\longleftmapsto}{\mbox{$\;\longleftarrow\!\mapstochar\;\,$}}
\newcommand{\longto}{\longrightarrow}
\newcommand{\hookto}{\hookrightarrow}
\newcommand{\toto}{\twoheadrightarrow}

\def\vv<#1>{\langle#1\rangle}

\newcommand{\pr}{\mbox{$\text{\up{pr}}$}}

\newcommand{\ev}{\mbox{$\text{\up{ev}}$}}

\providecommand{\det}{\mbox{$\text{\up{det}}\,$}}

\providecommand{\vol}{\mbox{$\text{\up{vol}}$}}

\newcommand{\dd}[2]{\mbox{$\frac{\partial #2}{\partial #1}$}}

\providecommand{\del}{\partial}

\newcommand{\om}{\omega}
\newcommand{\Om}{\Omega}

\newcommand{\eps}{\varepsilon}
\newcommand{\lam}{\lambda}

\newcommand{\by}[2]{\mbox{$\frac{#1}{#2}$}}
\newcommand{\cinf}{\mbox{$C^{\infty}$}}
\providecommand{\set}[1]{\mbox{$\{#1\}$}}

\newcommand{\X}{\mathfrak{X}}

\newcommand{\curv}{\mbox{$\textup{Curv}$}}

\newcommand{\ver}{\mbox{$\textup{Ver}$}}
\newcommand{\hor}{\mbox{$\textup{Hor}$}}

\newcommand{\momap}{momentum map\xspace}

\newcommand{\gu}{\mathfrak{g}}

\newcommand{\Ad}{\mbox{$\text{\upshape{Ad}}$}}

\newcommand{\GL}{\mbox{$\textup{GL}$}}
\newcommand{\gl}{\mbox{$\mathfrak{gl}$}}

\newcommand{\SO}{\mbox{$\textup{SO}$}}

\newcommand{\hl}{\mbox{$\textup{hl}^{\mathcal{A}}$}}

\newcommand{\Hamc}{\mbox{$\mathcal{H}_{\textup{c}}$}}
\newcommand{\Omnh}{\mbox{$\Om_{\textup{nh}}$}}
\newcommand{\Xnh}{\mbox{$X^{\textup{nh}}_{\mathcal{H}_{\textup{c}}}$}}

\newcommand{\todo}[1]{\vspace{5 mm}\par \noindent
\marginpar{\textsc{ToDo}} \framebox{\begin{minipage}[c]{0.95
\textwidth}\raggedright \tt #1 \end{minipage}}\vspace{5 mm}\par}

\begin{document}
\baselineskip=17pt
\titlerunning{Geometry of non-holonomic diffusion}

\title{Geometry of non-holonomic diffusion}

\author{Simon Hochgerner
\and 
Tudor S. Ratiu}
\date{}

\maketitle

\address{S. Hochgerner: Section de 
Math{\'e}matiques, Station 8,  
Ecole Polytechnique
  F{\'e}d{\'e}rale de Lausanne, 
1015 Lausanne,Switzerland; 
\email{simon.hochgerner@gmail.com}
\and
T.S. Ratiu: Section de 
Math{\'e}matiques, Station 8
and Bernoulli Center, Station 15, 
Ecole Polytechnique
  F{\'e}d{\'e}rale de Lausanne, 
1015 Lausanne, Switzerland;
\email{tudor.ratiu@epfl.ch}}

\begin{abstract}
We study stochastically perturbed non-holonomic systems from a
geometric point of view. In this setting, it turns out that the probabilistic properties
of the perturbed system are intimately linked to the geometry of the
constraint distribution. For $G$-Chaplygin
systems, this yields a stochastic criterion for the existence of
a smooth preserved measure.
As an application of our results we consider the motion planning problem
for the noisy two-wheeled robot and the noisy snakeboard.  

\end{abstract}

\maketitle

\tableofcontents

\section{Introduction} 

The goal of this paper is the study of stochastic non-holonomic systems.
This is a natural continuation of the work on stochastic Hamiltonian 
systems pioneered by Bismut~\cite{Bis81} and revitalized, brought up to 
date, and expanded by L\'{a}zaro-Cam\'{\i} and Ortega \cite{LO08} who also
connected it to symmetries, momentum maps, and reduction.

\subsection{Motivation and basic idea}
A non-holonomic system is, essentially, a rigid body together with a
set of constraints on the velocities. 
A prototypical example is the Chaplygin ball (\cite{C02}; 
for a modern treatment see \cite{D04} and \cite[Chapter
  6]{CuDuSn2010}). 
Here, the configuration
space is the direct product Lie group $G=\SO(3)\times\mathbb{R}^2$,
describing orientation and position of the ball,  
and the kinetic energy is specified by a
left-invariant metric $\mu$; there are 
two (non-integrable) velocity constraints so that the ball does
not slip, i.e., the point of contact of the ball and the plane has
zero velocity. Without constraints (which is clearly not the case in the
problem just presented), this 
would describe the motion of a rigid body in the plane, hence it would
be a Hamiltonian system. 

Stochastically perturbed versions of the latter setting (i.e., without constraints) 
have been
considered by L\'{a}zaro-Cam\'{\i} and Ortega \cite[Section~7.3]{LO08a}:
Let $h^0$ be the kinetic energy Hamiltonian of a left invariant metric
on the Lie group $G$, $\{y_i\}_i$ an orthonormal basis of the Lie algebra 
$\mathfrak{g}$ of $G$, $\{u_i\}_i$ its extension to a left invariant 
frame on $G$, and $h^i:T^*G\to\mathbb{R}$, $(q,p)\mapsto\vv<p,u_i(q)>$. 
Note that $h^i$ is the component $\vv<J^R,y_i>$ of the momentum map 
$J^R: T^\ast G \rightarrow \mathfrak{g}^\ast$ defined by the lift to 
$T^\ast G$ of right translation of $G$ on itself. The 
$\mathbb{R}\times \mathfrak{g}^\ast$-valued function $H=(h_0,h^i)$ on 
$T ^\ast G$ is left invariant. Following \cite{LO08,LO08a} and assuming 
that the perturbation is given by white noise, the stochastic rigid body 
is thus modeled by the Stratonovich equation 
\begin{equation}\label{e:mot-G}
 \delta\Gamma
 = X_{h^0}(\Gamma)\delta t + \sum X_{h^i}(\Gamma)\delta W^i,
\end{equation}
where $X_h$ denotes the Hamiltonian vector field of 
the function $h: T^*G\to\mathbb{R}$ and $W = \{W^i\}$ is Brownian 
motion in $\mathfrak{g}\cong\mathbb{R}^n$. A physical system modeled by 
this equation is that of a rigid body subject to small random impacts. 
Note that, since $u_i$ is auto-parallel for the Levi-Civita connection,
the equation $\delta\Gamma = \sum X_{h^i}(\Gamma)\delta W^i$ yields
the Hamiltonian construction of Brownian motion, as in \cite{LO08}. 

To pass to the nonholonomic setting, we note that the equations of
motion of the constrained (Chaplygin) ball can be encoded in the
vector field $PX_{h^0}$ where $P$ is the constraint force projection and 
is defined in \eqref{e:P} below. 
The effect of $P$ is to force the dynamics generated by $X_{h^0}$ to
satisfy the constraints. 
Thus, the idea of `the Hamiltonian construction of stochastic
non-holonomic systems' is to apply $P$ to \eqref{e:mot-G}. In fact, 
since $PX_{h^0}$ is nothing but the  non-holonomic 
vector field (see Section~\ref{sec:2}), we will focus on studying the 
effects of $P$ on the second term in equation~\eqref{e:mot-G}.
This yields non-holonomic constraints 
on the operator which is used to construct Brownian motion, 
thus leading to `constrained Brownian motion' described by 
\begin{equation}\label{e:mot-PG}
\delta\Gamma
 = \sum P(\Gamma)X_{h^i}(\Gamma)\delta W^i.
\end{equation}
As it stands, this equation has some problems. It depends very much on
the basis $\{u_i\}_i$ that was chosen in the definition of the
$h^i$. For example, since the no-slip constraints are actually
right invariant, one could have chosen a right invariant frame. But
then the Hamiltonian description of Brownian motion needs a correction
term involving the Levi-Civita connection of $\mu$. This approach has
been taken in \cite{H10}. However, 
the basis dependence implies that the generator of
\eqref{e:mot-PG} also changes when we pass to a different frame, and there
would be many natural choices depending on whether the frame should be
left or right invariant, adapted to the constraint distribution, or the
direct product structure of $G$, etc. 
Even if one ignores these issues, it is not clear what to do if
the configuration space is not parallelizable. For all these reasons
we transfer the construction to the bundle of orthonormal frames
itself. It is only then that the generator of the resulting
`constrained Brownian motion' is basis independent. 
This constrained Brownian motion has some interesting features:
\begin{itemize}
\item
To visualize it, we can think of a microscopic robot (or ball,
snakeboard, etc.) 
subject to 
molecular bombardment.
The robot thus experiences small 
impacts from all sides (isotropic in space) which force it to move around, but it still 
has to respect the constraints. 
\item
Now, it turns out, that the geometry of the constraints determines the 
probabilistic properties of the perturbed system. Indeed, if the 
constraints are integrable, then the robot's net drift will vanish. 
However, when the constraints are non-integrable and non-mechanical 
(which is the generic case) the Gaussian noise will induce a net drift on 
the robot. In Section~\ref{sec:4} we quantify this drift in terms of the 
geometry of the constraint distribution. \textit{Mechanical constraints}
are given, by definition as level sets of conserved quantities, such
as momentum maps. E.g., the constraints could be given by the
horizontal bundle of the mechanical connection, which is just
orthogonal to the vertical bundle in the case of a symmetry group action. 
\item
This leads to a dictionary between probabilistic aspects of the
perturbed system and classical properties of the original
(deterministic) non-holonomic system. 
See Theorem~\ref{thm:main} below for a preliminary statement of this
dictionary and Section~\ref{sec:4} for further details.
\end{itemize}

\subsection{Description of contents and results} 
Since this paper addresses both the 
geometric mechanics and the stochastic differential equations communities, 
we shall give the necessary background for all concepts and quote the
main results that are used later on. The paper is self contained. We 
briefly present the main results and the structure of the paper.

\subsubsection{Non-holonomic systems}
We start by recalling the necessary facts, concepts, and results of 
non-holonomic systems and their geometry. This includes a careful 
presentation of symmetries, reduction, and conditions for the existence 
of a (smooth) preserved measure. We will have to rephrase some of the 
existing results in view of applying them to our stochastic study
later on and develop the theory in the direction needed in subsequent
sections in the paper. 

Thus, we will have to give complete proofs not
only for some of the known results, due to our reformulation, but we
also need to establish new formulas. For example, the \textit{global} 
formula \eqref{Omega_expression} of the symplectic form on the tangent 
bundle given in terms of 
an underlying Riemannian metric on configuration space is new, as far 
as we know.
In \eqref{e:P} we introduce the above mentioned constraint force projection 
and explain its properties to prepare for Section~\ref{sec:4}. 
We also study Chaplygin systems, which are non-holonomic systems with
a particularly rich geometric structure, and the symmetry reduction of
such systems. 
One of
the main points of Section~\ref{sec:2} is the presentation of a
certain one-from $\beta$ which, according to Proposition~\ref{prop:pm},
characterizes the existence of a (smooth) preserved measure for a given
Chaplygin system.  
This result has been previously derived in \cite{CCLM02} but both our
proof and our interpretation of the relevant one-form $\beta$ are
different. In fact, our formulation of $\beta$ in \eqref{e:beta2} is a
prerequisite for Section~\ref{sec:4}.

\subsubsection{Stochastic dynamics on manifolds}
First, we recall some notions about manifold valued stochastic 
differential equations and diffusions from \cite{IW89,E89}.

Then we study symmetries of Stratonovich equations.
We consider a manifold $Q$ together with a proper action by a
Lie group $G$ and a diffusion $\Gamma^Q$ generated by a 
Stratonovich operator $\mathcal{S}$ from $T\mathbb{R}^{k+1}$ to
$TQ$ satisfying the equivariance relation \eqref{e:S-equiv}. 
In this setting, the Stratonovich operator does not (in general)
induce a Stratonovich operator on the base $Q/G$; however, the
diffusion $\Gamma^Q$ and its generator $A^Q$ are projectable to
$Q/G$. Thus, there is an induced diffusion $\Gamma^{Q/G}$ with induced 
generator $A^{Q/G}$  on the base space $Q/G$. See 
Theorem~\ref{prop:S-equiv}.

Two examples for this procedure of `equivariant reduction' 
are the Eells-Elworthy-Malliavin construction of Brownian motion 
(cf. equation~\eqref{e:SL})  on a 
Riemannian manifold and the stochastic Calogero-Moser systems (see 
\cite{H11}), as remarked in Subsection~\ref{sec:equiv-red}. 
In particular, we allow for non-free $G$-actions on $Q$ and hence $Q/G$
is, in general, not a smooth manifold but a stratified space. 
Thus, we extend the reduction theorem of \cite[Theorem~3.1]{LO08a} to 
the case when the Stratonovich operator on the total
space is not invariant but equivariant with respect to a symmetry
group action.

This naturally leads to the introduction, in 
Subsection~\ref{sec:equiv-diff}, of certain notions of equivariant diffusions,
previously studied in \cite{ELL04,ELL10}. The material of this
subsection will also be useful in Section~\ref{sec:5}.
In particular, we prove a mean reconstruction equation for diffusions in
principal bundles which is analogous to a concept by the same name in
mechanics  (see, e.g., \cite[\S4.3]{AbMa1978},
\cite[\S3]{MaMoRa1990}, \cite[Theorem~11.8]{Mon}) 
and uses that of \cite{ELL04,ELL10}.

\subsubsection{Non-holonomic diffusions}
This section contains the main results of the paper. We introduce 
constrained Brownian motion as motivated above.
This involves a careful analysis of the underlying geometry.
Then we study the  
generator and symmetry reduction of the resulting diffusion process. 
The reduction relies on Theorem~\ref{prop:S-equiv}.

The surprising fact in this regard, 
is that there is a very strong interrelation
of some probabilistic aspects of constrained Brownian motion and
certain deterministic properties of the original non-holonomic
system. A first instance of this relation is:

\begin{theorem}
Constrained
Brownian motion is a martingale with respect to the non-holonomic
connection on the configuration space. 
\end{theorem}

A second result yields a probabilistic characterization of the
existence of a preserved measure which is a very important concept in
the theory of non-holonomic systems (see
\cite{AKN02,B03,C02,EKMR04,HG09,H09,K92}):

\begin{theorem}\label{thm:main}
Let $(Q,\mathcal{D},L)$ be a $G$-Chaplygin system such that the base 
$M:=Q/G$ is compact. Let $\Gamma^M$ be the non-holonomic diffusion in 
$M$ associated to these data. Then the following are equivalent:
\begin{enumerate}[\up (1)]
\item
$(Q,\mathcal{D},L)$ has a (smooth) preserved measure;
\item
$\Gamma^M$ is time-reversible;
\item
$\Gamma^M$ has vanishing entropy production rate.
\end{enumerate}
\end{theorem}
The compactness assumption on $M$ is met in all classical examples
such as the Chaplygin ball or the two-wheeled robot.  
This theorem sums up some of the results of Sections~\ref{sec:4} and
\ref{sec:trd}, where also the relevant notions are introduced.

\subsubsection{Examples}
As examples, we consider the two-wheeled robot and the snakeboard. 
The former is $G$-Chaplygin and does (in general) not allow for a
preserved measure. The latter is not a Chaplygin system but does fit
the general set-up of Section~\ref{sec:4}. For both of these examples
we consider also the stochastic perturbation of deterministic
trajectory planning. This emphasizes the way in which the noise
couples with the constraints to produce a non-trivial drift
vector field (the emergence of which is at the heart of the geometry of
Section~\ref{sec:4}); this is in sharp contrast to stochastic Hamiltonian
systems. Indeed, the Hamiltonian analogue of non-holonomic reduction
is reduction at the $0$-level set of the standard cotangent bundle 
momentum map, which 
reduces Brownian motion to Brownian motion in the base with respect to
the induced metric. 
This is a manifestation of the idea that the amount by
  which a non-holonomic system differs from a Hamiltonian one can be
  measured by the amount by which the induced diffusion differs from
  Brownian motion -- and vice versa.

However, in the non-holonomic setting, the
constraints induce a drift giving rise to drifted Brownian motion on
the base space. This drift is quantified in Section~\ref{sec:4} and we
use it to make the perturbed motion follow a given curve \emph{on
average}. We show how the explicit form of the drift allows, in
principle, for a simple numerical implementation to solve such a motion
planning problem. It should be noted, though, that we have made no
attempt to study stability or convergence properties of the resulting
numerical algorithm. 
Similar problems have been treated, from a different perspective, in
the engineering literature; see \cite{ABG08,ZC04} and the references
therein.

\section{Non-holonomic systems}\label{sec:2}

We recall some facts about non-holonomic and, 
specifically, $G$-Chaplygin systems. 
Then we give a necessary and sufficient condition
for the existence of a preserved measure that is 
suitable for our
applications in Section~\ref{sec:4}.

A \textit{non-holonomic system} is a triple 
$(Q,\mathcal{D},\mathcal{L})$ consisting of a 
$n$-dimensional configuration manifold $Q$, a 
constraint distribution $\mathcal{D}\subset TQ$
which is smooth and of constant rank $r<n$ (i.e., it is
a vector subbundle of $TQ$ of rank $r$), and a smooth Lagrangian 
function $\mathcal{L}:TQ\to\mathbb{R}$. The dynamics 
of $(Q,\mathcal{D},\mathcal{L})$ are given by the
Lagrange-d'Alembert principle; see 
\cite{AKN02,BS93,B03,CMR01,C02,HG09,K92}. 
Throughout this paper, we assume that $\mathcal{L}$ is the kinetic energy 
of a Riemannian metric $\mu$ on $Q$.

\subsection{Almost Hamiltonian formulation}
Since $TQ\ni u_q\mapsto \mu(q)(u_q,\cdot) \in T^*Q$ 
is a vector bundle isomorphism covering the identity 
on $Q$, we shall identify the vector bundles $TQ$ 
with $T^*Q$. We follow \cite{BS93} to give 
an almost Hamiltonian description of the dynamics of 
$(Q,\mathcal{D},\mathcal{L})$. Let $\tau_Q: TQ\to Q$
be the tangent bundle projection and $\iota: \mathcal{D}\hookto TQ$ the inclusion. Define 
\begin{equation}
\label{def_c}
\mathcal{C}:= \left\{X_{u_q} \in T \mathcal{D}\mid
u_q \in \mathcal{D},\; T_{u_q}(\tau_Q\circ\iota)\left(
X_{u_q}\right) \in \mathcal{D} \right\}
= \left(T(\tau_Q\circ\iota)\right)^{-1}(\mathcal{D}). 
\end{equation}
In standard vector bundle charts of $TQ$ and $TTQ$, we write
$u_q$ as $(q, \dot{q})$ and
$X_{u_q}$ as $(q, \dot{q}, \delta q, \delta\dot{q})$, respectively.
Since $(\tau_Q \circ \iota)(q, \dot{q}) = q$, it
follows that $T(\tau_Q\circ\iota)(q, \dot{q}, \delta q, 
\delta\dot{q})=(q, \delta q)$ and hence $\mathcal{C}=
\{(q, \dot{q}, \delta q, \delta\dot{q})\mid 
(q, \dot{q}), (q, \delta q)
\in \mathcal{D}\}$, 
$\ker \left(T(\tau_Q\circ\iota)(q, \dot{q}, \cdot ,\cdot)
\right)=\{(q, \dot{q}, 0 ,\delta\dot{q})\mid 
\delta\dot{q} \in \mathbb{R}^n\}$. Thus 
$\mathcal{C}$ is a vector subbundle of 
$T\mathcal{D}$ of rank $2r$. 
(If $\mathcal{D}$ is the horizontal subbundle of a 
principal connection of some proper and free $G$-action 
on $Q$  then $\mathcal{C}$ is the
horizontal space of the tangent lifted $G$-action 
on $\mathcal{D}$. See \eqref{e:C-G} below.) 
According to \cite[Section~5]{BS93}  we have
\begin{equation}\label{e:C}
 \big(T(TQ)\big)|\mathcal{D}= \C\oplus\C^{\Om}
\end{equation}
where $\C^{\Omega}:=\{X_{u_q} \in T_{u_q}(TQ) \mid 
u_q \in \mathcal{D},\,\Omega(u_q)(X_{u_q}, Y_{u_q})=0, 
\forall Y_{u_q} \in \mathcal{C}\}$  is the 
$\Omega$-orthogonal complement of $\mathcal{C}$ in 
$\big(T(TQ)\big)|\mathcal{D}$; $\Omega$ denotes the 
canonical symplectic form on $TQ\cong T ^\ast Q$. We
will prove identity \eqref{e:C} later on, after the proof of
Proposition \ref{prop_Omega_formula}.

For reasons that will become clear in Section~\ref{sec:4}, we elaborate 
on \eqref{e:C}. We use the Levi-Civita connection $\nabla^{\mu}$ on 
$TQ\to Q$ to decompose $TTQ = \hor^{\mu}\oplus\ver(\tau_Q)$, where
$\ver(\tau_Q)=\ker(T\tau_Q: TTQ\to TQ)$
is the 
\textit{vertical} and $\operatorname{Hor}^\mu \subset 
TTQ$ is the \textit{horizontal subbundle}. Recall that a curve $v(t)$
in 
$TQ$ is \textit{horizontal}
if its covariant derivative $\frac{Dv(t)}{Dt}: =
\left.\frac{d}{ds}\right|_{s=0} \mathbb{P}_t^{t+s}v(t+s)$
vanishes; here $\mathbb{P}_t^{t+s}:T_{q(t+s)}Q 
\rightarrow T_{q(t)}Q$ is the parallel transport operator
of the Levi-Civita connection $\nabla^\mu$ and $q(t):=
\tau_Q(v(t))$. Alternatively, since $\frac{Dv(t)}{Dt}
= \nabla^\mu_{dq(t)/dt} v(t)$, or in coordinates,
$\frac{Dv^i(t)}{Dt}=\frac{dv^i(t)}{dt} + 
\Gamma_{jk}^i(q(t))\frac{dq^j(t)}{dt}v^k(t)$, the 
curve $v(t)$ is horizontal if and only if in any 
standard tangent bundle chart 
\begin{equation}
\label{horiz_cond}
\frac{dv^i(t)}{dt} + 
\Gamma_{jk}^i(q(t))\frac{dq^j(t)}{dt}v^k(t)=0.
\end{equation} 
A vector $X_{u_q} \in T_{u_q}TQ$ is called
\textit{horizontal} if it is tangent to a horizontal 
curve. The \textit{horizontal space} 
$\operatorname{Hor}^\mu_{u_q} \subset T_{u_q}TQ$ is the vector 
subspace formed by all horizontal vectors.

If $u_q = \dot{q}^i\frac{\partial}{\partial q^i} \in 
T_qQ$, the decomposition of a vector $X_{u_q} = 
A^i \frac{\partial}{\partial q ^i} + 
B^i\frac{\partial}{\partial\dot{q}^i} \in T_{u_q}TQ$ in its horizontal and vertical part is
\begin{equation}
\label{hv_dec}
A^i\frac{\partial}{\partial q^i} + 
B^i\frac{\partial}{\partial \dot{q}^i} = 
\left(A^i\frac{\partial}{\partial q^i} -
\Gamma^i_{jk}\dot{q}^jA^k
\frac{\partial}{\partial \dot{q}^i}\right) +
\left(\Gamma^i_{jk}\dot{q}^jA^k +B^i\right)
\frac{\partial}{\partial \dot{q}^i}\,.
\end{equation}
Indeed, since $T_{u_q}\tau_Q\left(
R^i\frac{\partial}{\partial q^i} + 
S^i \frac{\partial}{\partial \dot{q}^i}\right) = 
R^i\frac{\partial}{\partial q^i}$ it follows that
\[
 \ker T_{u_q} \tau_Q = \left\{\left. S^i\frac{\partial}{\partial
  \dot{q}^i} \,\right|\, S^i 
\in \mathbb{R}
\right\}
\]
which shows
that the second summand in \eqref{hv_dec} is vertical.
The first summand is horizontal since it verifies the
horizontality condition \eqref{horiz_cond} (with 
$v^i=\dot{q}^i$, $A^i=\frac{dq^i}{dt}$, and 
$\frac{dv^i}{dt}= - \Gamma^i_{jk}\dot{q}^jA^k$). In
particular, note that $T_{u_q} \tau_Q: \mathcal{C}_{u_q}
\cap \operatorname{Hor}^\mu_{u_q} \rightarrow 
\mathcal{D}_q$ is an isomorphism: 
$A^i\frac{\partial}{\partial q^i} -
\Gamma^i_{jk}\dot{q}^jA^k
\frac{\partial}{\partial \dot{q}^i} \in 
\operatorname{Hor}^\mu_{u_q}$ maps to the given vector 
$A^i\frac{\partial}{\partial q^i} \in \mathcal{D}_q$. 
Similarly $T_{u_q}\tau_Q: \hor^{\mu}_{u_q}\to T_qQ$ 
is an isomorphism. Its inverse is the horizontal lift 
mapping which is often written as a map 
$\operatorname{hl}^{\mu}: TQ\times_QTQ\cong\hor^{\mu}$,
$(u_q,v_q)\mapsto(T_{u_q}\tau_Q|\hor^{\mu}_{u_q})^{-1}
(v_q)$. Interpreting $\operatorname{pr}_1: TQ \times_Q 
TQ \rightarrow TQ$ as a vector bundle over $TQ$ with 
base the first factor, makes $\operatorname{hl}^{\mu}: 
TQ\times_QTQ\stackrel{\sim}\rightarrow \hor^{\mu}$ into a 
vector bundle isomorphism covering the identity on $TQ$.

Let $K: \operatorname{Ver}( \tau_Q) \rightarrow TQ \times_Q TQ$ be 
the inverse to the vertical lift mapping 
$\textup{vl}:TQ\times_Q TQ\stackrel{\sim}\longrightarrow 
\ver(\tau_Q)$ defined by $\operatorname{vl}(u_q,v_q):= 
\left.\frac{d}{dt}\right|_{t=0}(u_q+tv_q)$, for all
$u_q, v_q \in T_qQ$. In standard coordinates, 
$K(q, \dot{q}, 0, \delta\dot{q}) = (q,\dot{q},q,\delta\dot{q})$. 
 In particular, $K(X_{u_q}) \in T_qQ$.
In addition, $T\tau_Q: \hor^{\mu}\to TQ$
and $K: \ver(\tau_Q)\to TQ$ restricted to each fiber
over $TQ$ are linear isomorphisms. Let $P_{\rm hor}$ and 
$P_{\rm ver}$ denote the horizontal and vertical 
projections associated to $\hor^{\mu}$. By abuse of notation, 
we sometimes write $K$ 
also for $K\circ P_{\rm ver}: TTQ\to\ver(\tau_Q)\to TQ$.
We have thus the vector bundle isomorphism over 
$\mathcal{D}$
\begin{align}\label{first_c}
\mathcal{C} 
 &\stackrel{\sim}\longrightarrow 
 (\mathcal{D}\times_Q\mathcal{D}) \oplus \ker T(\tau_Q\circ\iota), \\
\label{first_c_direct}
 X_{u_q}&\longmapsto
 \Big(u_q,T_{u_q}\tau_Q\left(X_{u_q}\right),
     K\left(P_{\rm ver}\left(X_{u_q}\right)\right)
 \Big),\\
 \label{first_c_inverse}
 \textup{hl}^{\mu}_{u_q}(v_q) + \operatorname{vl}(u_q,w_q)
 & \longleftmapsto 
 (u_q, v_q,w_q),
\end{align}
where we regard $\mathcal{D}\times_Q \mathcal{D}\ni(u_q,v_q) \mapsto
u_q \in \mathcal{D}$ 
as a vector
bundle over $\mathcal{D}$.
Notice also that $T\mathcal{D}\supset \ker T(\tau_Q\circ\iota)
= \bigsqcup_{(q,u)\in\mathcal{D}}\textup{vl}_{(q,u)}\mathcal{D}_q$.

\begin{proposition}
\label{prop_Omega_formula}
The canonical symplectic form $\Omega \in \Omega^2(TQ)$
has the expression
\begin{equation}
\label{Omega_expression}
\Omega(u_q)\left(X_{u_q},Y_{u_q}\right)
 = \mu(q)(T_{u_q}\tau_Q(X_{u_q}),K(Y_{u_q}))
 -\mu(q)(T_{u_q}\tau_Q(Y_{u_q}),K(X_{u_q})),
\end{equation}
for any $q \in Q$, $u_q \in T_qQ$, $X_{u_q},Y_{u_q} \in 
T_{u_q}(TQ)$.
\end{proposition}

\begin{proof}
In an arbitrary standard tangent bundle chart, we have
\begin{equation}
\label{riemannian_omega}
\Omega = \frac{\partial \mu_{ik}}{ \partial q^j}\dot{q}^k
\mathbf{d}q^i \wedge \mathbf{d}q^j +
\mu_{ij}\mathbf{d}q^i \wedge \mathbf{d}\dot{q}^j,
\end{equation}
where the Riemannian metric is written as $\mu=
\mu_{ij} \mathbf{d}q^i \otimes\mathbf{d}q^j$, with 
$\mu_{ij} = \mu_{ji}$. Thus, if 
\[
X_{u_q} = A ^i\frac{\partial}{\partial q^i} + 
B^i\frac{\partial}{\partial \dot{q}^i}\,, \qquad 
Y_{u_q} = C^i\frac{\partial}{\partial q^i} + 
D^i\frac{\partial}{\partial \dot{q}^i}\,,
\]
 we get 
\begin{equation}
\label{formula}
\Omega(u_q)(X_{u_q}, Y_{u_q})=
\frac{\partial \mu_{ik}}{ \partial q^j}\dot{q}^k 
(A^iC^j - A^jC^i) + \mu_{ij}(A^iD^j - C^iB^j).
\end{equation}
On the other hand, $T_{u_q}\tau_Q(X_{u_q}) =
A^i\frac{\partial}{\partial q^i}$, 
$T_{u_q}\tau_Q(Y_{u_q}) = C^i\frac{\partial}{\partial q^i}$ and
\begin{align*}
K\left(A^i\frac{\partial}{\partial q^i} + 
B^i\frac{\partial}{\partial \dot{q}^i}\right)&=
\left(\Gamma^i_{jk}\dot{q}^jA^k +B^i\right)
\frac{\partial}{\partial q^i}, \\ 
K\left(C^i\frac{\partial}{\partial q^i} + 
D^i\frac{\partial}{\partial \dot{q}^i}\right)&=
\left(\Gamma^i_{jk}\dot{q}^jC^k +D^i\right)
\frac{\partial}{\partial q^i}
\end{align*}
by \eqref{hv_dec} and the definition of $K$. Therefore,
\begin{align*}
&\mu(q)(T_{u_q}\tau_Q(X_{u_q}),K(Y_{u_q}))
 -\mu(q)(T_{u_q}\tau_Q(Y_{u_q}),K(X_{u_q}))\\
 & \qquad= \mu_{ij}\Gamma^j_{rk}\dot{q}^r(A^iC^k-A^kC^i)
 + \mu_{ij}(A^iD^j-C^iB^j)\\
 & \qquad = \mu_{ij}\frac{1}{2}\mu^{js}\left(
 \frac{\partial \mu_{sk}}{\partial q^r} +
  \frac{\partial \mu_{sr}}{\partial q^k} - 
   \frac{\partial \mu_{rk}}{\partial q^s} \right)
   \dot{q}^r(A^iC^k-A^kC^i)
 + \mu_{ij}(A^iD^j-C^iB^j)\\
 & \qquad = \frac{1}{2}\left(
 \frac{\partial \mu_{ik}}{\partial q^r} +
  \frac{\partial \mu_{ir}}{\partial q^k} - 
   \frac{\partial \mu_{rk}}{\partial q^i} \right)
   \dot{q}^r(A^iC^k-A^kC^i)
 + \mu_{ij}(A^iD^j-C^iB^j)\\
& \qquad= \frac{1}{2}
\left(
 \frac{\partial \mu_{ik}}{\partial q^r} -
  \frac{\partial \mu_{ir}}{\partial q^k} - 
   \frac{\partial \mu_{rk}}{\partial q^i} \right)
   \dot{q}^r(A^iC^k-A^kC^i)+
 \frac{\partial \mu_{ir}}{\partial q^k}
  \dot{q}^r(A^iC^k-A^kC^i)\\
  & \qquad \qquad 
 + \mu_{ij}(A^iD^j-C^iB^j)\\
 &\qquad = -\Gamma_{ik}^sg_{rs}
 \dot{q}^r(A^iC^k-A^kC^i)+
 \frac{\partial \mu_{ir}}{\partial q^k}
  \dot{q}^r(A^iC^k-A^kC^i)
 + \mu_{ij}(A^iD^j-C^iB^j)\\
 &\qquad  = \frac{\partial \mu_{ir}}{\partial q^k}
  \dot{q}^r(A^iC^k-A^kC^i)
 + \mu_{ij}(A^iD^j-C^iB^j)
\end{align*}
because $\Gamma_{ik}^s$ is symmetric and
$(A^iC^k-A^kC^i)$ is skew-symmetric in $(i,k)$. 
However, this expression coincides with
\eqref{formula} which proves \eqref{Omega_expression}.  
\end{proof}

Thus by \eqref{Omega_expression} we get
\begin{align}
\label{second_c}
\mathcal{C}^\Omega_{u_q} 
& = \left\{X_{u_q}\in T_{u_q}(TQ) \mid u_q \in 
\mathcal{D}_q,\right. \\
& \qquad \left. \mu(q)(T_{u_q}\tau_Q(X_{u_q}),K(Y_{u_q}))
 -\mu(q)(T_{u_q}\tau_Q(Y_{u_q}),K(X_{u_q}))=0, \;
 \forall Y_{u_q} \in \mathcal{C}_{u_q}\right\} 
 \nonumber \\
 &= \left\{X_{u_q}\in T_{u_q}(TQ) \mid u_q \in 
\mathcal{D}_q,\; K(P_{\rm ver}(X_{u_q})) \in 
\mathcal{D}_q^\perp,\;
T_{u_q}\tau_Q(P_{\rm hor}(X_{u_q})) \in \mathcal{D}_q^\perp
\right\} \nonumber \\
&\cong (\mathcal{D}\times_Q\mathcal{D}^{\perp}) 
\oplus \bigsqcup_{u_q\in\mathcal{D}}
\operatorname{vl}_{u_q}(\mathcal{D}^{\bot}) \nonumber 
\end{align} 
since $K, T_{u_q} \tau_Q: 
\mathcal{C}_{u_q}\rightarrow \mathcal{D}_q$ 
are surjective, where 
$\mathcal{D}^{\perp}\subset TQ$ is the $\mu$-orthogonal of $\mathcal{D}$ and the vector bundle isomorphism in the last line of \eqref{second_c}
is given by $X_{u_q} \mapsto (u_q, T_{u_q} \tau_Q(X_{u_q}),
P_{\rm ver}(X_{u_q}))$. This expression of $\mathcal{C}^\Omega$
and \eqref{def_c} show that $\mathcal{C}\cap 
\mathcal{C}^ \Omega=\{0\}$ which proves \eqref{e:C}.

In particular, if 
\begin{equation}\label{e:P}
 P: (T(TQ))|\mathcal{D}=\C\oplus\C^{\Om}\to\C
\end{equation} 
is the projection along $\C^{\Om}$ and
$\Pi: TQ=\mathcal{D}\oplus\mathcal{D}^{\bot}\to\mathcal{D}$ is the orthogonal projection
then it follows that 
\begin{equation}\label{e:Pi}
 T(\tau_Q\circ \iota)\circ P = \Pi\circ T(\tau\circ \iota).
\end{equation}
Indeed,  using the above description of $\C$ and 
$\C^{\Om}$, this follows immediately by decomposing 
$\big(T(TQ)\big)|\D$ into its
horizontal and vertical parts. 

Let $\Ham$ be the kinetic energy Hamiltonian on $TQ$ 
which we regard as the Legendre
transform of $\mathcal{L}$. Then the dynamics of the 
non-holonomic system
$(Q,\mathcal{D},\mathcal{L})$ are given by the vector field 
\[
 X_{\mathcal{H}}^{\mathcal{C}} := PX_{\mathcal{H}} 
 \in\X(\mathcal{D})
\]
where $X_{\mathcal{H}}$ is the Hamiltonian vector field 
of $\Ham$. More generally, for a function $f\in\cinf(TQ)$ 
we regard $X_{f}^{\mathcal{C}} := PX_{f}\in
\mathfrak{X}(\mathcal{D})$  as the non-holonomic
vector field of $f$. Let $\Omega^{\mathcal{C}}$ denote 
the fiberwise restriction of $\iota^*\Omega$ to
$\C\times\C$.  Then \eqref{e:C} implies that 
$\Omega^{\mathcal{C}}$ is non-degenerate and 
we may rewrite the defining equation for 
$X^{\mathcal{C}}_f$ as
\[
\mathbf{i}_{X^{\mathcal{C}}_f}\Omega^{\mathcal{C}}
 = (\mathbf{d}f)^{\mathcal{C}}
\]
where $(\mathbf{d}f)^{\mathcal{C}}$ is the fiberwise 
restriction of $\iota^*(\mathbf{d}f)$ to $\C$.

\subsection{$G$-Chaplygin systems}
\label{sec_Chaplygin}
Now we shall consider the case when the non-holonomic 
system is invariant under a group action such that the 
constraints are given by a principal bundle connection. 
A \textit{$G$-Chaplygin system} consists of a Riemannian 
configuration space $(Q,\mu)$, a Lie group $G$ with Lie algebra $\mathfrak{g}$ which 
acts freely and properly on $(Q,\mu)$
by isometries, and a principal bundle connection 
$\mathcal{A}\in \Omega^1(Q; \mathfrak{g})$ on
$\pi: Q\toto Q/G =: M$. For $\xi\in \mathfrak{g}$ denote
by $\xi_Q \in \mathfrak{X}(Q)$ the infinitesimal generator defined by
\[
\xi_Q(q): = \left.\frac{d}{dt}\right|_{t=0}\exp (t \xi)
\cdot q
\]
for all $q \in Q$, where $\exp: \mathfrak{g}\rightarrow 
G $ is the exponential map.

The Lagrangian of this system is 
the kinetic energy $\mathcal{L} := 
\frac{1}{2}\|\cdot\|_{\mu}^2$. It is also assumed that 
the constraint distribution is the horizontal subbundle
of the connection $\mathcal{A}$, i.e., 
$\mathcal{D} := \ker\mathcal{A}\subset TQ$. Thus the 
system $(Q,\mathcal{D},\mathcal{L})$ is a
non-holonomic system and the dynamics are determined by 
the Lagrange-d'Alembert equations; see
\cite{AKN02,BS93,B03,CMR01,C02,HG09,K92}. 
It is \textit{not} assumed that $\mathcal{D}$ is
orthogonal to the vertical space $\ker T\pi$.

Since $\D$ is the horizontal subbundle, it is invariant 
with respect to the tangent lifted $G$-action on $TQ$. 
Thus we obtain a principal $G$-fiber bundle
$\D\toto\D/G = TM$. This bundle carries an induced connection
$\iota^*\tau^*\A$, where $\iota: \D\hookto TQ$ is
the inclusion and $\tau: TQ\to Q$ is the tangent
bundle projection. Its associated horizontal bundle is 
\begin{equation}\label{e:C-G}
\ker(\tau\circ\iota)^*\mathcal{A}=
\{u_q \in T\mathcal{D}\mid T(\tau\circ\iota)u_q\in\ker\mathcal{A}=\mathcal{D}\}
 =
 \mathcal{C}.
\end{equation}
Let $\mu_0$ denote the induced Riemannian metric on 
$M:=Q/G$. Then the isomorphism
\[
T_q \pi: \left(\D_q, \mu(q)|\D_q \right)
\rightarrow \left(T_{\pi(q)}M, \mu_0(\pi(q))\right)
\]
is an isometry for the indicated inner products for 
all $q \in Q$.

\subsection{The non-holonomic correction} In order to 
carry out non-holonomic reduction we need to introduce
a two-form on $TM$ induced by the momentum map and the
curvature $\operatorname{Curv}^\mathcal{A} \in 
\Omega^2(Q; \mathfrak{g})$ of the connection 
$\mathcal{A}$. As we shall see in the next subsection,
this form is the correction that one needs to subtract
from the canonical symplectic form in order to give
an almost Hamiltonian formulation of the reduced
non-holonomic system. To define this form, we need 
three ingredients:
\begin{itemize}
\item[(i)] The \textit{adjoint bundle}:
Let $G$ act on $Q \times \mathfrak{g}$ by the (free and
proper) action given by $g \cdot (q, \xi) := 
(g\cdot q,\operatorname{Ad}_g \xi)$, for all $g \in G$, 
$q \in Q$, $\xi\in \mathfrak{g}$, and let 
$\widetilde{\mathfrak{g}}:= Q \times_G \mathfrak{g} 
= (Q \times \mathfrak{g})/G$ be the orbit space. 
Elements of $\widetilde{\mathfrak{g}}$ are denoted
by $[q, \xi]_G$. The projection $\rho:
\widetilde{\mathfrak{g}}\ni [q, \xi]_G \mapsto \pi(q) 
\in M$ defines the \textit{adjoint vector bundle} 
whose fibers are Lie algebras.

\item[(ii)] The \textit{curvature on the base}:
$\operatorname{Curv}^\mathcal{A} \in 
\Omega^2(Q; \mathfrak{g})$ naturally induces a 
two-form on $\operatorname{Curv}_0^\mathcal{A} \in 
\Omega^2(M; \widetilde{\mathfrak{g}})$ on the base $M$ 
with values in the adjoint bundle 
$\widetilde{\mathfrak{g}}$ by
\[
\operatorname{Curv}_0^\mathcal{A}(\pi(q)) \left(
T_q \pi(u_q), T_q \pi(v_q) \right): = 
\left[q, \operatorname{Curv}^ \mathcal{A}(q)
(u_q,v_q) \right]_G
\]
for all $q \in Q$, $u_q, v_q \in T_qQ$.

\item[(iii)] The \textit{momentum map of the tangent 
lifted $G$-action}: $\mathbf{J}_G: TQ \rightarrow 
\mathfrak{g}^\ast$ is defined by 
$\left\langle\mathbf{J}_G(u_q), \xi\right\rangle
= \mu(q)(u_q, \xi_Q(q))$ 
for all $\xi\in \mathfrak{g}$, $u_q \in TQ$, and is
equivariant.
\end{itemize}

To get a grip on the non-holonomic correction two-form, 
we begin describing it if $G$ is a commutative group. 
Then the adjoint bundle is trivial:
$\rho: \widetilde{\mathfrak{g}} = M \times \mathfrak{g}
\rightarrow M$ is the projection on the first factor. 
Thus, $\operatorname{Curv}_0^ \mathcal{A} \in 
\Omega^2(M; \mathfrak{g})$ and we define the
\textit{non-holonomic correction two-form} $\Xi\in 
\Omega^2(TM)$ by $\Xi : = \left\langle 
 \mathbf{J}_G\circ\operatorname{hl}^{\mathcal{A}},
 \tau_M^*\operatorname{Curv}_0^{\mathcal{A}}
\right\rangle$, that is,
\begin{equation}
\label{correction_form_abelian}
\Xi(u_x) \left(X_{u_x}, Y_{u_x}\right) : = 
\left\langle \mathbf{J}_G\left(
{\operatorname{hl}^{\mathcal{A}}}_q(u_x) \right),
\operatorname{Curv}_0^{\mathcal{A}}(x)\left(
 T_{u_x} \tau_M(X_{u_x}), T_{u_x} \tau_M(Y_{u_x}) \right)
\right\rangle
\end{equation}
for all $x \in M$, $u_x \in T_xM$, $X_{u_x}, Y_{u_x}
\in T_{u_x}(TM)$,
where ${\operatorname{hl}^\mathcal{A}}_q:=\left(
T_q \pi| \mathcal{D}_q\right)^{-1}: T_xM
\rightarrow \mathcal{D}_q\subset TQ$, $x = \pi(q)$, is the 
horizontal lift operator associated to the connection
$\mathcal{A}$. The pairing on the right hand side of
this formula is between $\mathfrak{g}^\ast$ and 
$\mathfrak{g}$. The right hand side of this formula
seems to depend on $q \in Q$. However, this is not the
case because the horizontal lifts at two distinct
points in $Q$ are related by a group element and 
the momentum map is invariant under the $G$-action
(since $G$ is commutative).

As stated, this formula does not make sense for 
general Lie groups because the momentum map is 
$\mathfrak{g}^\ast$-valued and the curvature on the
base is $\widetilde{\mathfrak{g}}$-valued so the pairing
makes no sense. However, the idea for the general 
formula is based on \eqref{correction_form_abelian}.
We define $\Xi\in \Omega^2(TM)$ by
\begin{multline}\label{e:JK}
\Xi(u_x)(X_{u_x},Y_{u_x})\\
:=\vv<\mathbf{J}_G(\hl_q(u_x)),
\operatorname{Curv}^{\mathcal{A}}(q)\left(
{\operatorname{hl}^\mathcal{A}}_q\left(
T_{u_x}\tau_M(X_{u_x})\right),
{\operatorname{hl}^\mathcal{A}}_q\left(
T_{u_x}\tau_M(Y_{u_x}))\right)\right)>
\end{multline}
for $X_{(x,u)},Y_{(x,u)}\in T_{(x,u)}(TM)$ and 
$q\in\pi^{-1}(x)$; since both entries in this pairing 
are $G$-equivariant the ambiguity cancels out, that is, 
the right hand side in \eqref{e:JK} does not depend 
on $q$ but only on $\pi(q)=x$. 

Due to the importance of this formula we make a few
additional comments. Recall that the momentum map 
$\mathbf{J}_G:TQ \rightarrow \mathfrak{g}^\ast$ is
equivariant with respect to the coadjoint action on 
$\mathfrak{g}^\ast$. The tangent lifted $G$-action
restricts to an action on $\D\subset TQ$; indeed 
$\D=\ker\A$ is the horizontal subbundle and is hence  
$G$-invariant. Corresponding to the $G$-principal
bundle projection $\D\toto\D/G=TM$ there is a natural 
connection which is induced from the connection 
$\A$ on $Q\toto Q/G$, namely $\iota^*\tau^*\A$ where
$\iota: D\hookto TQ$ is the inclusion and 
$\tau: TQ\to Q$ is the tangent bundle projection. 
The curvature of $\iota^*\tau^*\A$ is 
$\iota^*\tau^*\curv^{\mathcal{A}}$ which is equivariant:
$l_g^*\iota^*\tau^*\curv^{\mathcal{A}} =
\Ad_g\circ\left(\iota^*\tau^*\curv^{\mathcal{A}}\right)$, 
where $l_g: \D\to\D$ is the action of $g\in G$.
Thus the two-form $\vv<\mathbf{J}_G,\iota^*\tau_Q^*
\curv^{\mathcal{A}}>$ defines a $G$-invariant 
two-form on $\D$. 
This two-form is, moreover, horizontal: since
$\iota^*\tau^*\curv^{\mathcal{A}}$ is a curvature form 
on $\D\toto\D/G$ it vanishes upon insertion of vertical 
vectors, whence the same holds also for
$\vv<\mathbf{J}_G,\iota^*\tau^*\curv^{\mathcal{A}}>$. 
Thus the two-form $\vv<\mathbf{J}_G,\iota^*\tau_Q^*
\curv^{\mathcal{A}}>$ is basic and hence drops to 
a well-defined two-form $\Xi$ on $\D/G=TM$. Implementing
the computations suggested above gives \eqref{e:JK}.

\subsection{Non-holonomic reduction} Identify $TQ$ with 
$T^*Q$ by the metric $\mu$ and $TM$ with $T^*M$ by the
metric $\mu_0$. Consider the orbit projection map
\[
 T\pi|\mathcal{D}: \mathcal{D}\toto \mathcal{D}/G = TM.
\]
We may also associate a fiberwise inverse to this
mapping which is given by the horizontal lift mapping
$\hl: Q\times_M TM\to \mathcal{D}$ associated to 
$\mathcal{A}$. The following statements are proved
in \cite{BS93,EKMR04,HG09}.

\begin{proposition}[Non-holonomic reduction]
\label{prop:comp}
The following hold.
\begin{enumerate}[\up (1)]
\item
$\Om^{\mathcal{C}}$ descends to a
non-degenerate two-form $\Omnh$ on $TM$.
\item
$\Omnh = \Om_M - \Xi\in \Omega^2(TM)$, 
where $\Omega_M = -\mathbf{d}\theta_M$ is the canonical symplectic form on $TM$ and $\Xi$ is the non-holonomic
correction two-form given by \eqref{e:JK}.
 
\item
Let $h: TQ\to\mathbb{R}$ be $G$-invariant. Then
the vector field
$X^{\mathcal{C}}_h$ is $T\pi|\mathcal{D}$-related
to the vector field $X^{\textup{nh}}_{h_0}$ on $TM$ defined by
\[
\mathbf{i}_{X^{\rm nh}_{h_0}}\Omnh = \mathbf{d}h_0
\] 
where $h_0:TM \rightarrow \mathbb{R}$ is the induced 
Hamiltonian.
\end{enumerate}
\end{proposition}

In general, $\Omnh$ is an almost symplectic form, that 
is, it is non-degenerate and non-closed.  
We will denote the reduced Hamiltonian  by $\Hamc$ and 
refer to the almost Hamiltonian system $(TM,\Omnh,\Hamc)$ 
as the \textit{reduced data}. The identity 
$\Omnh = \Om_M - \Xi$ appears for the first time, albeit not
completely explicitly,
in \cite{BS93}.
A proof  using moving frames is given in
\cite{EKMR04}
where it is also called the ``$\vv<J,K>$-formula". 
A different proof
following the above outline
is contained in \cite[Prop~2.2]{HG09}.

\subsection{The preserved measure}
Does $(TM,\Omnh,\Hamc)$ possess a preserved measure? 
This is an important question since it says something 
about the possible existence of asymptotic equilibria 
and also plays a prominent role in the theory of 
integration of non-holonomic systems. Correspondingly,
this topic is touched upon in all of 
\cite{AKN02,B03,C02,EKMR04,HG09,H09,K92}.
In \cite{CCLM02}, a necessary and sufficient condition 
for the existence of a preserved measure in terms of 
local coordinates on the base manifold $M$ is given. 
We derive derive below an equivalent formulation of
this result which is more closely adapted to the 
Riemannian structure on $M$.  This point of view will
then be exploited in Section~\ref{sec:4} in the stochastic context.

For brevity we will denote 
$X=X^{\textup{nh}}_{\mathcal{H}_{\textup{c}}}$ in 
this subsection. Let $\Omega^m$, $m=\dim M$, be the 
Liouville volume on $TM$.  Then \textit{there is a
preserved measure for the flow of $X$ if and only if 
there is a strictly positive function 
$\N: M\to\mathbb{R}$ such that $(\N\circ \tau_M)\Omega^m$ 
is  preserved, that is,}
\begin{equation}\label{e:N}
\boldsymbol{\pounds}_X((\N\circ \tau_M)\Omega^m) = 0;
\end{equation}
see \cite{CCLM02} for a proof. In such a case $\N$ is called the
\textit{density} 
of the preserved measure
with respect to the Liouville volume. As shown in
\cite[Remark~7.4]{CCLM02}, it suffices to consider density functions on $M$. 

To reformulate condition \eqref{e:N} we want to use the 
fact that $(M,\mu_0)$ is a Riemannian manifold. Hence 
we equip $TM$ with the Sasaki metric $\sigma$ 
associated to $\mu_0$ (see, e.g., \cite{GK02}),
\begin{equation}
\label{sasaki_metric}
 \sigma(u_x)(X_{u_x},Y_{u_x}):= 
 \mu_0(x)\left(T_x\tau_M(X_{u_x}),T_x\tau_M(Y_{u_x})\right)
 + \mu_0(x)\left(K_M(X_{u_x}), K_M(Y_{u_x})\right)
\end{equation}
for all $X_{u_x},Y_{u_x}\in T_{u_x}(TM)$, where 
$\tau_M:TM \rightarrow M$ is the tangent bundle 
projection and $K_M:\operatorname{Ver}(\tau_M) 
\rightarrow TM \times _M TM$ is the inverse of the 
vertical lift map $\operatorname{vl}_M:TM \times _M TM 
\stackrel{\sim}\rightarrow \operatorname{Ver}(\tau_M)$;
note, in particular that $K_M(X_{u_x}) \in T_xM$. 
We recall some of the key properties of the Sasaki
metric; see \cite{GK02} for proofs.
\begin{itemize}
\item[(i)] The Sasaki metric $\sigma$ is the unique 
Riemannian metric  on $TM$ such that $\tau_M: 
(TM, \sigma) \rightarrow (M, \mu_0)$ is a Riemannian 
submersion, that is, the isomorphism 
\[
T_{u_m} \tau_M: 
 \left(\left(\ker T_{u_m} \tau_M\right)^{\perp_\sigma}
 = \hor_{u_m}, \sigma(u_m) \right) 
 \longrightarrow 
 (T_mM, \mu_0(m))
\]
is an isometry (for the indicated inner products) for 
all $u_m \in TTM$, where $\perp_ \sigma$ denotes the perpendicular 
relative to the Sasaki inner product $\sigma(u_m)$ on $T_{u_m}(TM)$.
\item[(ii)] $\hor$ and $\ver$ are $\sigma$-perpendicular
complements of each other: $\operatorname{Hor}= \operatorname{Ver}^{\perp_ \sigma}$.
\item[(iii)] The vertical lift map $\textup{vl}: TM\times_M
  TM\to\ver\subset TTM$ 
is an isometry 
of vector bundles over $TM$, thinking
of the projection onto the first factor 
$\operatorname{pr}_1:TM\times_M TM \rightarrow TM$
as a vector bundle over $TM$ and $\mu_0$ as a vector
bundle metric.
\end{itemize}

Given a vector field $X$
on $M$ we shall denote its horizontal lift relative to
the Riemannian metric $\mu_0$ by $X^h\in\X(TM,\hor)$ 
and its vertical lift by $X^v\in\X(TM,\ver)$.

\begin{lemma}[The non-holonomic vector field]
\label{lm:nh-vf}
If $X_0=X_{\mathcal{H}_{\rm c}}=\Om_M^{-1}(\mathbf{d}\Hamc) \in
\mathfrak{X}(TM)$ is 
the standard Hamiltonian
vector field,
$X = \Om_{\textup{nh}}^{-1}(\mathbf{d}\Hamc)\in\mathfrak{X}(TM)$ is the
non-holonomic vector field,  
and $\{u_1,\dots,u_m\}$ is a local orthonormal frame on $M$,
then the following hold:
\begin{align}
 T_{u_x}\tau_M(X(u_x)) 
  &= T_{u_x}\tau_M(X_0(u_x)) 
   = u_x
  \textup{ for all }
  u_x\in TM,\\
 P_{\textup{ver}}\Big(X-X_0\Big)(u_x)
&=
  -\sum_{i=1}^m\Xi(u_x)\Big(X(u_x),u_i^h(u_x)\Big)\, u_i^v(u_x)\\
  &=
  -\sum_{i=1}^m\vv<(\mathbf{J}_G\circ\hl_q)(u_x),
       (\curv^{\mathcal{A}}_q\circ\wedge^2\hl_q)(u_x,u_i(x))>\, u_i^v(u_x)\notag
\end{align} 
where the second equation holds locally in the domain of definition of
the given frame,
and is well-defined independently of the choice of $q\in\pi^{-1}(x)$. 
\end{lemma}

\begin{proof}
We begin by noting that $\Om_M(X_0,Y) = 
\mathbf{d}\Hamc(Y) =
\Omnh(X,Y)$ for all $Y\in\X(TM)$. Hence by
Propositions~\ref{prop:comp} and \ref{prop_Omega_formula}
\begin{align*}
 \Xi(u_x)\Big(X(u_x),Y(u_x)\Big)
 &=
 \Om_M(X-X_0,Y)_{u_x}\\
 &=
 \mu_0(x)(T_{u_x}\tau_M((X-X_0)(u_x)),K_{u_x}Y(u_x))\\
 &\phantom{==}
 -
 \mu_0(x)(T_{u_x}\tau_M(Y(u_x)),K_{u_x}(X-X_0)(u_x)). 
\end{align*}
This implies, firstly, that 
$T_{u_x}\tau_M(X(u_x)) 
 = T_{u_x}\tau_M(X_0(u_x)) = u_x$
since $\Hamc$ is the kinetic energy Hamiltonian of the induced metric
$\mu_0$. Secondly, since
$K_{u_x}\left(u_i^v(u_x)\right)=u_i(x)=T_{u_x}\tau_M\left(u_i^h(u_x)\right)$,
we find locally
\begin{align*}
 P_{\textup{ver}}(X-X_0)(u_x)
 &=
  \sum\sigma\Big((X-X_0)(u_x),u_i^v(u_x)\Big) \,u_i^v(u_x)\\
 &=
  \sum\mu_0\Big(K_{u_x}(X-X_0)(u_x),u_i(u_x)\Big) \,u_i^v(u_x)\\
 &=
  -\sum\Om_M(X-X_0,u_i^h)_{u_x} \,u_i^v(u_x)\\
 &=
  -\sum\Xi(u_x)\Big(X,u_i^h\Big) \,u_i^h(u_x)\\ 
 &=
  -\sum\vv<(\mathbf{J}_G\circ\hl)_q(u_x),
       (\curv^{\mathcal{A}}_q\circ\wedge^2\hl_q)(u_x,u_i(x))> \,u_i^v(u_x)
\end{align*}
where we have used 
$T_{u_x}\tau_M\left(X(u_x)\right) = u_x$ in the last line.
\end{proof}

Let $\vol_{\sigma}$ be the volume form on $TM$ induced 
by the Riemannian metric $\sigma$. We shall prove
the following formula: 
\begin{equation}
\label{sigma_volume}
\vol_{\sigma}=\by{1}{m!}\Omega^m.
\end{equation}
Indeed by \eqref{riemannian_omega}, denoting by 
$\mathfrak{S}_m$ the permutation group of 
$\{1, \ldots, m\}$, we have
\begin{align*}
\Omega^m
&=
\left(\frac{\partial \mu_{ik}}{ \partial q^j}\dot{q}^k
\mathbf{d}q^i \wedge \mathbf{d}q^j +
\mu_{ij}\mathbf{d}q^i \wedge \mathbf{d}\dot{q}^j\right)^m\\
&=
(\mu_{ij}\mathbf{d}q^i \wedge \mathbf{d}\dot{q}^j)^m\\
&=
\sum_{\pi\in \mathfrak{S}_m}
\mu_{1\pi(1)}\cdots\mu_{m\pi(m)}
\mathbf{d}q^1\wedge \mathbf{d}\dot{q}^{\pi(1)}\wedge\ldots\wedge \mathbf{d}q^m\wedge 
\mathbf{d}\dot{q}^{\pi(m)}\\
&=
m!\left(\sum_{\pi\in \mathfrak{S}_m}
  (-1)^{\operatorname{sign}\pi}
  \mu_{1\pi(1)}\cdots\mu_{m\pi(m)}
\right)
\mathbf{d}q^1\wedge \mathbf{d}\dot{q}^1\wedge
\ldots\wedge \mathbf{d}q^m\wedge\mathbf{d}\dot{q}^m\\
&=
m!\det(\mu_{ij})
\mathbf{d}q^1\wedge \mathbf{d}\dot{q}^1\wedge
\ldots\wedge \mathbf{d}q^m\wedge\mathbf{d}\dot{q}^m.
\end{align*}
On the other hand, in the coordinates 
$(v^1, v^2, \ldots, v^{2m-1}, v^{2m})$ of $TM$, where
$v^{2i-1}=q^i$ and $v^{2i}=\dot{q}^i$ for $i=1,\ldots,m$,
we have by the usual formula of the Riemannian volume,
$\operatorname{vol}_ \sigma= 
\sqrt{\operatorname{det}(\sigma_{IJ})}
\mathbf{d}q^1\wedge \mathbf{d}\dot{q}^1\wedge
\ldots\wedge \mathbf{d}q^m\wedge\mathbf{d}\dot{q}^m$,
where $\sigma_{IJ}:=\sigma\left(
\frac{\partial}{\partial v^I},
\frac{\partial}{\partial v^J}\right)$. Since, by the 
definition \eqref{sasaki_metric} of the Sasaki metric 
we have $\mu_{ij}= \sigma\left(
\frac{\partial}{\partial q^i},
\frac{\partial}{\partial q^j}\right) =
\sigma\left(
\frac{\partial}{\partial \dot{q}^i},
\frac{\partial}{\partial \dot{q}^j}\right)$ and
$\sigma\left(
\frac{\partial}{\partial q^i},
\frac{\partial}{\partial \dot{q}^j}\right)=0$, it 
follows that the matrix $(\sigma_{IJ})$ is of the form
\[
(\sigma_{IJ}) = P
\begin{pmatrix}
 (\mu_{ij}) && 0 \\
 0 && (\mu_{ij})
\end{pmatrix}P^{-1}
\]
for a permutation matrix $P$. Therefore, 
$\sqrt{\operatorname{det}(\sigma_{IJ})} = 
\operatorname{det}(\mu_{ij})$ which proves 
\eqref{sigma_volume}.

By \eqref{sigma_volume}, condition
\eqref{e:N} holds if and only if
\begin{equation}\label{e:N1}
 \boldsymbol{\pounds}_{X}((\N\circ \tau_M) \vol_{\sigma})  =  0
 \;\Longleftrightarrow \;
 \left\langle 
  \mathbf{d}(\log\N\circ \tau_M),X
 \right\rangle 
 +
 \operatorname{div}_{{\rm vol}_{\sigma}}X  
 =  0.
\end{equation}

Let us define
\[
 L(TM) 
 :=
 \{
  l\in\cinf(TM): l_x:=l\mid T_xM: 
  T_xM\to\R
  \textup{ is linear for all }x\in M
 \}
\]
and consider the prescription
$\Phi: L(TM)\to\Om^1(M)$, $\Phi(l)(u_x) = l_x(u_x)$, 
$u_x \in T_xM$, which is an
isomorphism of $\cinf(M)$-modules. 

\begin{lemma}\label{lem:div}
The following statements hold.
\begin{enumerate}[\up (1)]
\item
$\textup{div}_{\textup{vol}_{\sigma}}X\in L(TM)$. 
\item
Let $\{u_i\mid i=1, \ldots, m\}$ denote a local orthonormal frame on $M$. Then
\begin{align}
\label{div_formula}
 \operatorname{div}_{{\rm vol}_{\sigma}}X (u_x)
 &= - \sum_{i=1}^m \Xi(u_i(x))\Big(X(u_x), u_i^h(u_x)\Big)\\
 &= - \sum_{i=1}^m 
 \left\langle 
  (\mathbf{J}_G\circ\hl_q)(u_i(x)), (\curv_q^{\mathcal{A}}\circ\wedge^2\hl_q)(u_x,u_i(x))
 \right\rangle. \nonumber 
\end{align}
\end{enumerate}
\end{lemma}

\begin{proof}
Clearly (2) implies (1) so we shall prove (2) below.

We use the Levi-Civita connection $\nabla^{\mu_0}$ to 
split $TTM=\hor\oplus\ver$ where $\ver = \ker T(\tau: 
TM\to M)$.  Given a vector field $X$
on $M$ we shall, as before, denote its horizontal lift by
$X^h\in\X(TM,\hor)$ 
and its vertical
lift by $X^v\in\X(TM,\ver)$.

Let $\{u_i\in\X(M)\mid i=1,\dots,\dim M\}$, be a local orthonormal frame
for $TM$. Then $\{(u_i^h,u_i^v)\mid i=1,\dots,\dim M\}$ is a local
orthonormal frame for $TTM$ with respect to $\sigma$. 
By Lemma~\ref{lm:nh-vf}, if $X_0=X_{\mathcal{H}_{\rm c}}
=\Om_M^{-1}(\mathbf{d}\Hamc) \in \mathfrak{X}(TM)$
is the standard Hamiltonian vector field,
then we can locally express $X$ as
\[
 X = 
 X_0 - \sum\vv<(\mathbf{J}_G\circ\hl_q)(u_x),
       (\curv^{\mathcal{A}}_q\circ\wedge^2\hl_q)(u_x,u_i(x))> \,u_i^v(u_x).
\]
Notice that $X_0$ preserves $\Om^m$ whence $\textup{div}_{\textup{vol}_{\sigma}}X_0 =
0$.
According to, e.g.,
\cite[Proposition~7.2]{GK02},
it is true that 
$\nabla^{\sigma}_{u_i^h}u_i^h = (\nabla^{\mu_0}_{u_i}u_i)^h$
and $\nabla^{\sigma}_{u_i^v}u_i^v = 0$ where $\nabla^{\sigma}$ is the
Levi-Civita connection of $\sigma$. 
Therefore,
sometimes
suppressing the base point $u_x$ for readability,
\begin{align*}
 \textup{div}_{\textup{vol}_{\sigma}}X
 &= \textup{Tr}\,\nabla^{\sigma}_. X
  = \sum_{i=1}^m\big(\sigma(\nabla^{\sigma}_{u_i^h} X,u_i^h) 
     + \sigma(\nabla^{\sigma}_{u_i^v} X,u_i^v)\big)\\
 &= \sum_{i=1}^m\big(
     - \sigma(X,\nabla^{\sigma}_{u_i^h}u_i^h)
     + u_i^h\sigma(X,u_i^h)
     - \sigma(X,\nabla^{\sigma}_{u_i^v}u_i^v)
     + u_i^v\sigma(X,u_i^v)
     \big)\\
 &= \sum_{i=1}^m\big(
     - \sigma(X_0,\nabla^{\sigma}_{u_i^h}u_i^h)
     + u_i^h\sigma(X_0,u_i^h)
     + u_i^v\sigma(X_0,u_i^v)\\
 & \qquad
 - u_i^v \vv<(\mathbf{J}_G\circ\hl_q)(u_x), 
(\curv^{\mathcal{A}}_q\circ\wedge^2\hl_q)(u_x,u_i(x))>
     \big)\\
 &= \textup{div}_{\textup{vol}_{\sigma}}X_0 
     - \sum_{i=1}^m \left\langle 
             (\mathbf{J}_G\circ\hl)_q(u_i(x)), 
(\curv_q^{\mathcal{A}}\circ\wedge^2\hl_q)(u,u_i(x))
            \right\rangle\\
&= -\sum_{i=1}^m\left\langle 
(\mathbf{J}_G\circ\hl)_q(u_i(x)), 
(\curv_q^{\mathcal{A}}\circ\wedge^2\hl_q)(u,u_i(x))
   \right\rangle.
\end{align*}
To see that the formula does not depend on the particular choice of
$q\in\pi^{-1}(x)$ one uses equivariance of the involved expressions
together with the observation that any $G$-ambiguity cancels out in
the pairing.
\end{proof}

Therefore,
$\textup{div}_{\textup{vol}_{\sigma}}X$ can be turned into a one-form
on $M$ through the canonical isomorphism $\Phi: L(TM)\to\Om^1(M)$.
Let us define
\begin{equation}\label{e:beta1}
 \beta 
 := -\Phi\Big(\textup{div}_{\textup{vol}_{\sigma}}X\Big)
 \in\Omega^1(M)
\end{equation}
for the non-holonomic vector field $X=X^{\textup{nh}}_{\mathcal{H}_{\rm c}}$.

\begin{proposition}\label{prop:pm}
The system $(TM,\Omnh,\Hamc)$ admits a preserved measure 
if and only if $\beta\in\Omega^1(M)$ is exact. If 
$\beta=\mathbf{d}F$ for some function $F$ on
$M$ then $\N=e^F$ is the density of the preserved measure 
for the Liouville volume.
\end{proposition} 

\begin{proof}
By Lemma~\ref{lm:nh-vf}, 
$\vv<\mathbf{d}(\log\N \circ \tau_M),X>(u_x) = \mathbf{d}(\log\N)(u_x)$ for $u_x\in
TM$ and hence $\Phi\Big(\vv<\mathbf{d}(\log\N \circ \tau_M),X>\Big) =
\mathbf{d}(\log\N)$. Now by \eqref{e:N1} a preserved measure $(\N\circ
\tau_M)\Om_M^m$ 
exists
if and only if
\[
 \mathbf{d}(\log\N)
 = 
 \Phi\Big(\vv<\mathbf{d}(\log\N),X>\Big) 
 =
 -\Phi(\textup{div}_{\textup{vol}_{\sigma}}X)
 =
 \beta,
\]
i.e., $\beta\in\Om^1(M)$ is exact.
\end{proof} 

As stated above, this result is proved in 
\cite[Theorem~7.5]{CCLM02} but
our interpretation of the form $\beta$ is slightly different.
The formula 
\begin{align}\label{e:beta2}
 \beta(x)(u_x) &= 
  \sum_{i=1}^m\Xi(u_i(x))\Big(X(u_x), u_i^h(u_x)\Big)\\
  &=
 \sum_{i=1}^m 
 \left\langle 
  (\mathbf{J}_G\circ\hl_q)(u_i(x)), (\curv_q^{\mathcal{A}}\circ\wedge^2\hl_q)(u_x,u_i(x))
 \right\rangle, \nonumber 
\end{align}
$u_x \in T_xM$, in a local orthonormal frame $\{u_1, \ldots, u_m\}$
will be useful in Section~\ref{sec:4} below.

\section{Stochastic dynamics on manifolds}\label{sec:3}

\subsection{Diffusions on manifolds}
This subsection is a review of some necessary definitions and results which
are all contained in the books \cite{IW89,E89}.

A diffusion is a continuous stochastic process which has the strong
Markov property. This is a concept which can be formulated in any
topological space.

\subsubsection{Diffusion processes}
Let $X$ be a locally compact topological space with one-point
compactification $\dot{X}=X\cup\set{\infty}$
and endow $\dot{X}$ with its Borel $\sigma$-algebra $\mathcal{B}(X)$.
Define
$W(X)$ to be the set of all maps $w: [0,\infty)\to\dot{X}$ such that
there is a $\zeta(w)\in [0,\infty]$ satisfying
\begin{enumerate}[\up(1)]
\item
$w(t)\in X$ for all $t\in[0,\zeta(w))$ and $w: [0,\zeta(w))\to X$ is
    continuous;
\item
$w(t)=\infty$ for all $t\ge\zeta(w)$.
\end{enumerate}

Let $l\in\mathbb{N}$, $0\le t_1<\ldots<t_l\in\R_+$, 
$A\subset\Pi_{i=1}^l\dot{X}$ a Borel set, and
consider the evaluation mapping $\ev(t_1,\dots,t_l):
W(X)\to\Pi_{i=1}^l\dot{X}$, $w\mapsto(w(t_1),\dots,w(t_l))$. Then 
\[
 S = \ev(t_1,\dots,t_l)^{-1}(A)
\]
is called a \emph{Borel cylinder set} in $W(X)$. If $t\ge0$ and
$t_l\le t$ then $S$ is a \emph{Borel cylinder set up to time} $t$.

The set $W(X)$ is equipped with the $\sigma$-algebra $\mathcal{B}(W(X))$ generated by all
Borel cylinder sets in $W(X)$. This $\sigma$-algebra has a natural
filtration given by the family of 
\[
 (\mathcal{B}_t(W(X)))_{t\ge0}
\] 
which are the $\sigma$-algebras generated by
Borel cylinder sets up to time $t$. 

A family of probabilities $(P_x)_{x\in\dot{X}}$ on
$(W(X),\mathcal{B}(W(X)))$ is said to be a
\emph{system of diffusion measures} on
$(W(X),\mathcal{B}(W(X)),\mathcal{B}_t(W(X)))$ if it has the strong
Markov property, the definition of which we will give shortly.

A $(\mathcal{B}_t(W(X)))_t$-\emph{stopping time}
is a random variable $\tau: W(X)\to\dot{\R}_+ = \R_+\cup\{\infty\}$ such that 
$\set{w \in W(X) \mid\tau(w)\leq
  t}\in\mathcal{B}_t(W(X))$ for all $t\in\R_+$. 
  
For $s\in\R_+$ we define the \textit{time shift operator} 
\begin{equation}\label{e:time-shift}
 \Sigma_s: W(X)\longto W(X),\quad
 w\longmapsto(\Sigma_sw: t\mapsto w(s+t)).
\end{equation}
A family of probabilities $(P_x)_{x\in\dot{X}}$ on
$(W(X),\mathcal{B}(W(X)))$ satisfies the \emph{strong Markov property}
if, for all $x\in X$, $(\mathcal{B}_t(W(X)))_t$-stopping times $\tau$,
bounded $\mathcal{B}_{\tau}(W(X))\times\mathcal{B}(W(X))$-measurable 
functions $F: W(X)\times W(X)\to\R$, and $s\in\R_+$, we have
\begin{equation}\label{e:strong-markov}
 \int_{\tau(w)<\infty}F(w,\Sigma_{\tau(w)}w)\,P_x(dw)
 =
 \int_{\tau(u)<\infty}\left(
   \int_{W(X)}F(u,w)\,P_{u(\tau(u))}(dw)
\right)\,P_x(du).
\end{equation}
See \cite[p.~249]{Str}.

Let $(\Om,\F,P)$ be a probability space and $\Gamma:
\Om\times\mathbb{R}_+\to\dot{X}$ a map. 
Then $\Gamma$ is called a \emph{stochastic process} 
if $\Gamma_t:\Om\to\dot{X}$, $\om\mapsto\Gamma_t(\om)$, is a random variable for all $t\in\R_+$. 
Define $\check{\Gamma}: \om\mapsto(t\mapsto\Gamma_t(\om))$. 
Then $\Gamma$ is said to be a \emph{continuous stochastic process} in $X$ if
$\check{\Gamma}: (\Om,\F)\to(W(X),\mathcal{B}(W(X)))$ is a random
variable. (Below, when $X$ is a manifold, we will only be
  dealing with continuous processes.) Note that, for all 
$\omega\in \Omega$, the paths $[0, \zeta(\check{\Gamma}(\om)))\ni t 
\mapsto \Gamma_t(\omega)\in X$ are continuous; the map $\zeta:
W(X)\to[0,\infty]$ was part of the definition of $W(X)$.

The \textit{law} of $\Gamma$ is, by definition, the push-forward probability
$\check{\Gamma}_*P$ on $(W(X),\mathcal{B}(W(X)))$, i.e., $\check{\Gamma}_*P(S) =
P(\check{\Gamma}^{-1}(S))$ for all $S\in\mathcal{B}(W(X))$. 

The process
$\Gamma: \Om\times\R_+\to\dot{X}$ defined on the probability space
$(\Om,\F,P)$ 
is a \emph{diffusion} in $X$ if it is a continuous process
and
there is a system
 of diffusion measures $(P_x)_{x\in\dot{X}}$ such that 
$\check{\Gamma}_*P = P_{\mu}$ as probability laws on
 $(W(X),\mathcal{B}(W(X)))$;
here 
\[
 P_{\mu}(S) := \int_{\dot{X}}P_x(S)\mu(dx)\quad
 \textup{ for all }\quad
 S\in\mathcal{B}(W(X))
\]
and $\mu = (\Gamma_0)_*P: \mathcal{B}(X)\to[0,1]$ is the initial
distribution of $\Gamma$. 

A diffusion $\Gamma$ in $X$ with associated system of diffusion
measures $(P_x)_x$  is said to be
\emph{generated} by a linear operator $A$ on
the Banach space of continuous functions
$C(\dot{X})$ with domain of definition $\A\subset C(\dot{X})$ if, 
for all $x\in X$, $t\ge0$,
and $f\in\A$, the stochastic process $M_t^f: W(X)\to\R$,
\[
 M_t^f(w)
 :=
 f(w(t)) - f(w(0)) - \int_0^t(Af)(w(s))\,ds,
\]
is a $P_x$-martingale on  $(W(X),\mathcal{B}(W(X)))$ for the
filtration $(\mathcal{B}_t(W(X)))_{t\ge0}$. In this case, $A$ is called
the \emph{generator} of $\Gamma$. See 
\cite[Defs.~IV.5.3 and IV.6.2]{IW89}. The definition of a martingale is recalled below.

Let $(\Om,\F,P)$ be a probability space. A family
$(\F_t)_{t\in\mathbb{R_+}}$ of sub-$\sigma$-algebras $\F_t\subset\F$ 
is called a \emph{reference family} if it
is increasing, i.e., $\F_t\subset\F_s$ for $0\le t\le s$, and
right-continuous, i.e., $\cap_{\epsilon>0}\F_{t+\epsilon}=\F_t$ for
all $t\in\R_+$. Whenever we mention $(\F_t)$ we will suppress the
index set $\R_+$, 
tacitly assume
that it is a reference family, and refer to it as the filtration of
$\F$ so that $(\Om,\F,(\F_t),P)$ becomes a filtered probability space.

A stochastic process $M: \Om\times\R_+\to\dot{\R}_+$ 
is called a \textit{martingale} on $(\Om,\F,(\F_t),P)$ 
if the following conditions are met:
\begin{enumerate}[\up (1)]
\item
$M_t: \Om\to\R$ is integrable for all $t\in\R_+$;
\item
$M_t: \Om\to\R$ is $\F_t$-measurable for all $t\in\R_+$, 
i.e., $M$ is $(\F_t)$-\textit{adapted};
\item
$E[M_t|\F_s] = M_s$ for all $t\ge s \ge 0$, i.e.,  
$E[(M_t-M_s) \chi_F]=0$ for all 
$t\geq s\geq 0$ and all $F \in \mathcal{F}_s$, where $\chi_F$ is the 
characteristic function of the set $F$.
\end{enumerate}

If $W: \Om\times\R_+\to\R^k$ is
the (a fortiori continuous) 
diffusion
defined on the filtered probability space
$(\Om,\F,(\F_t),P)$ 
with initial condition $W_0=0$ a.s.\ 
and with generator $\by{1}{2}\Delta =
\by{1}{2}\sum\del_i\del_i$, then $W$ is called
a$(\F_t)$-\textit{adapted Brownian motion}. See 
\cite[Example~IV.5.2]{IW89} or
\cite[Remark~7.1.23]{Str} for this characterization 
of Brownian motion. Below, we will be concerned with 
Brownian motion on a Riemannian manifold and then this 
aforementioned characterization will be taken
to be the definition of Brownian motion.

\subsubsection{Diffusions via Stratonovich equations}
Let $(\Om,\F,(\F_t),P)$ be a filtered probability space as above and
suppose
now that $X=Q$ is a manifold. 
From now on, all stochastic processes will be assumed to be
continuous.

If $N$ is manifold then a
\emph{Stratonovich operator} $\mathcal{S}$ from $TN$ to $TQ$ is a
section of $T^*N\otimes TQ\to N\times Q$. 
Equivalently, we can view $\mathcal{S}$ as a smooth map 
$\mathcal{S}: Q\times TN\to TQ$ 
which is linear in the fibers and covers  the identity on
$Q$. 
Let $X_0,X_1,\dots,X_k$ be
vector fields on $Q$ and define the associated Stratonovich 
operator $\mathcal{S}: Q\times T\mathbb{R}^{k+1}\longto TQ$ by 
\[
 \mathcal{S}(x,w,w'):=\sum_{i=0}^k X_i(x)\vv<e_i,w'>,
\]
where $x \in Q$, $w\in \mathbb{R}^{k+1}$, $(w,w') \in 
T_w \mathbb{R}^{k+1} = \{w\} \times \mathbb{R}^{k+1}$, $\{e_i\mid
i=0,1,\ldots, k\}$ 
is the orthonormal standard basis in 
$\mathbb{R}^{k+1}$, and $\left\langle\,, \right\rangle$ is the standard inner
product in $\mathbb{R}^{k+1}$.
We note that the number $k$ is not related to the dimension of $Q$.

Consider the stochastic process $Y: \Om\times\mathbb{R}_+\to\mathbb{R}^{k+1}$,
$(t,\om)\mapsto(t,W_t(\om))$ where $W$ denotes 
$(\F_t)$-adapted Brownian motion in $\mathbb{R}^{k}$.

We will be concerned with Stratonovich equations of the form 
\begin{equation}\label{e:S1}
 \delta\Gamma = \mathcal{S}(Y,\Gamma)\delta Y;
\end{equation}
a continuous $(\F_t)$-adapted process $\Gamma: \Om\times\mathbb{R}_+\to Q$  is
called a \textit{solution} to 
\eqref{e:S1} if there is a $(\F_t)$-adapted Brownian motion $W = (W^i)$ 
in $\mathbb{R}^{k}$ such that, in the
Stratonovich sense,
\begin{equation}\label{e:def-of-sol}
 f(\Gamma_t) - f(\Gamma_0)
 = \int_0^t(X_0f)(\Gamma_s)ds
   + \sum_{i=1}^k\int_0^t(X_if)(\Gamma_s)\delta W_s^i
\end{equation}
for all smooth functions $f\in C^{\infty}(Q)$ with 
compact support.

A few comments are in order. The definition of 
\eqref{e:S1} is \eqref{e:def-of-sol}. The second 
integral in \eqref{e:def-of-sol} is the Stratonovich
integral (signaled by saying that the equation 
is to hold ``in the Stratonovich sense"). We will not
go into the elaborate definition and construction of
the Stratonovich integral here and refer to 
\cite{Oks07,IW89,E89} for the definition and
an in depth study of the Stratonovich integral.

For the readers more familiar with It\^o calculus we
add the following remarks. It\^o integrals are defined 
by a Riemann sum approximation with the rather essential 
difference that, in the sum, one evaluates 
the integrand 
\emph{at the left end-points} of the
partition intervals.  
The Stratonovich 
integral, on the other hand, can be obtained by evaluating
the integrand \emph{at the mid-points} of the 
sub-intervals. While It\^o integrals give rise to 
a new transformation rule (the It\^o formula),
transformed Stratonovich integrals obey the same 
change of variables formula as Riemann integrals. This is
the essential reason why, on manifolds, It\^o 
calculus is replaced by Stratonovich calculus. 
Concretely, the Stratonovich integral is characterized
in \cite[Thm.~III.1.4]{IW89} by
\[
\int_0^tY\,\delta X =
\textup{l.i.p.}_{|\Lambda|\to0}\sum_{i=1}^n\by{Y(t_i)
-Y(T_{i-1})}{2}(X(t_i)-X(t_{i-1}))
\]
where $\Lambda$ is a partition $0=t_0<t_1<\ldots<t_n=t$ 
with maximal step size $|\Lambda|$ and $\textup{l.i.p.}$ 
is ``limit in probability"; here $X$, $Y$ are 
quasimartingales (a general class of semi-martingales).
We do not go into more details of the definition of
the Stratonovich integral here and refer the reader
to the above mentioned books.

Suppose $\Gamma$ is a solution to \eqref{e:S1} such that $\Gamma_0=x$
a.s.\ and $\Gamma$ satisfies \eqref{e:def-of-sol} with respect to an
$\R^k$-valued Brownian motion $W$ defined on a filtered probability
space $(\Om,\F,(\F_t),P)$.
Then we will write $\Gamma=\Gamma^{x,W}$ to remember
these data.
The explosion time $\zeta$ of a solution $\Gamma^{x,W}$ is a stopping
time on $(\Om,\F,(\F_t))$ with the following property: 
the path $\Gamma^{x,W}_{[0,T]}(\om)$ is
contained in $Q$ for all $T<\zeta(\om)$ but
if
$\zeta(\om)<\infty$ then
$\Gamma^{x,W}_{[0,\zeta(\om))}(\om)$ is not contained in any compact
  subset of $Q$. 
The following is a partial account of \cite[Theorems~V.1.1 and
  V.1.2]{IW89} and \cite[Theorem~(7.21)]{E89} that is sufficient for
our purposes. 

\begin{theorem}
Let the assumptions be as above and consider equation \eqref{e:S1}.
\begin{enumerate}[\up (1)]
\item
For each initial condition, $\Gamma_0=x$ a.s.,  
and continuous $(\F_t)$-adapted Brownian motion $W$,
a solution $\Gamma^{x,W}$ exists and is unique up to explosion time. 
\item
Let $P_x := \check{\Gamma}^{x,W}_*P$. Then $P_x$ is independent of $W$
and $(P_x)$ is a system of diffusion measures generated by the second
order differential operator
\begin{equation}\label{e:A-thm}
 A = X_0 + \by{1}{2}\sum_{i=1}^k X_ix_i.
\end{equation}
which acts on the space $\cinf(Q)_0$ of smooth functions with compact support. 
\end{enumerate}
\end{theorem}

Assume that $Q$ is endowed with a linear connection 
$\nabla: \mathfrak{X}(Q) \times \mathfrak{X}(Q)
\rightarrow \mathfrak{X}(Q)$, where $\mathfrak{X}(Q)$ denotes the
Lie algebra of smooth vector fields on $Q$.

If $f\in C^{\infty}(Q)$,
its \textit{Hessian} is defined by 
$\operatorname{Hess}(f):=\nabla\nabla f$, i.e., 
$\operatorname{Hess}(f)(X,Y)= X(Y(f)) -(\nabla_XY)(f)$
for any $X,Y \in \mathfrak{X}(Q)$.
The Hessian is bilinear in $X$ and $Y$ but is not symmetric, unless 
$\nabla$ is torsion-free. 

Let $\Gamma$ be a diffusion in $Q$ with generator $A$. 
Then the \emph{drift} of
$\Gamma$ with respect to $\nabla$ is defined to be the first order
part of $A$ which is determined by $\nabla$. If $A$ is of the form
\eqref{e:A-thm} then this is $X_0+\by{1}{2}\sum\nabla_{X_i}X_i$.

According to \cite[Theorem~7.31]{E89}, \emph{the $A$-diffusion
$\Gamma$ is a martingale 
in $(Q,\nabla)$ if and only if 
$A$
is purely second order with
respect to $\nabla$}, i.e., the $\nabla$-drift vanishes. 
In \cite{E89} this is stated for torsion-free connections but it is
noted that one can use the same definition for connections with
torsion.

If $(Q,\mu)$ is a Riemannian manifold then an 
$A$-diffusion is called \textit{Brownian motion} if 
$A = \frac{1}{2}\Delta$ where $\Delta
:=
\operatorname{div} \operatorname{grad} = 
- \delta\mathbf{d}$ is the
metric Laplacian; see \cite[Def.~V.4.2]{IW89} or \cite[Def.~5.16]{E89}.

To construct Brownian motion in $(Q,\mu)$, we need the
principal connection 
\[
 \omega \in\Omega^1\left(\mathfrak{F}; 
 \mathfrak{so}(d)\right)
\]
on the orthonormal frame bundle 
$\rho: \mathfrak{F}\to Q$ 
over 
$(Q,\mu)$, uniquely induced by the Levi-Civita connection 
$\nabla^\mu$ on $Q$ (whose Christoffel symbols in a chart
are denoted by $\Gamma_{kl}^i$). We recall its construction 
and basic properties. Let $u\in\mathfrak{F}$ with base point 
$\rho(u)=q\in Q$. Then we define the horizontal bundle as 
$\hor^{\om}=\bigsqcup_{u\in\mathcal{F}}\hor^{\om}_u$, where
the horizontal space at $u$, a vector subspace of 
$T_u \mathfrak{F}$,  is given by
\begin{equation}\label{e:hor-om}
\hor_u^{\om}:= T_q\sigma\left(T_qQ\right);
\end{equation}
here $\sigma$ is local section of $\rho:\mathfrak{F}\to Q$ such that 
$\sigma(q) = u$ and $\nabla^{\mu}_X\sigma_i = 0$ for all $X\in T_qQ$ and 
local vector fields $\sigma_i:=\sigma(e_i)$
with $\{e_1,\dots,e_d\}$ being the standard basis in $\R^d$.   
We may express \eqref{e:hor-om} in
local coordinates $(q^i,u^i_j)$ defined on a bundle coordinate
patch $U\times V\hookto\mathfrak{F}$ as 
\begin{equation}\label{e:hor-om-coord}
 \hor^{\om}_u
 =
 \left.\left\{ A^k\frac{\del}{\del q^k} 
    -   \Gamma^i_{kl}u^l_j A^k\frac{\del}{\del u^i_j} \,\right|\, A ^k\in \mathbb{R}
 \right\}.
\end{equation}
Restricting $\om$ to this coordinate patch, it can be written 
as $\pr_2 + \om_{\textup{loc}}: 
T(U\times V)\to\mathfrak{so}(d)$, where 
$\om_{\textup{loc}} =
(\om_i^j)\in\Omega^1(U;\mathfrak{so}(d))$.
It follows that, in terms of the Christoffel symbols,
\begin{equation}\label{e:om-coord}
 \om^i_j = \Gamma^i_{kl}u^l_jdq^k.
\end{equation}
Thus, the local expression of the horizontal lift defined by
$\omega$ is
${\textup{hl}^{\om}}_u(\frac{\del}{\del q^k})
= 
\frac{\del}{\del q^k} 
   -  \Gamma^i_{kl}u^l_j \frac{\del}{\del u^i_j}$.

An orthonormal frame $u\in\mathfrak{F}$ can
be regarded as an isometry $u: \mathbb{R}^d\to T_{\rho(u)}Q$,
where $d=\dim Q$. Define
the canonical horizontal vector fields $L_i\in
\mathfrak{X}(\mathfrak{F},\hor^{\om})$, $i=1,\dots,d$, by
\begin{equation}\label{e:L}
 L_i(u) = {\textup{hl}^{\omega}}_u(u(e_i)),
\end{equation}
where $\textup{hl}^{\omega}: \mathfrak{X}(Q)\to\mathfrak{X}(\mathfrak{F})$ 
is the horizontal lift map of $\omega$.
If $(W^i)$ is Brownian motion in $\mathbb{R}^d$ and $\Gamma$ solves the
Stratonovich equation
\begin{equation}\label{e:SL}
 \delta\Gamma = \sum L_i(\Gamma)\delta W^i
\end{equation}
then $\rho\circ\Gamma$ is a diffusion in $Q$ with generator
$\by{1}{2}\Delta^{\mu}$, that is, a Brownian motion. This is explained
in \cite[Chapter~V.4]{IW89} and follows also from
Theorem~\ref{prop:S-equiv} below; the essential observation in
this context is that the Stratonovich operator $\mathcal{S}(u,w,w') =
\mathcal{S}(u,w)w' = \sum L_i(u)\vv<e_i,w'>$ 
of \eqref{e:SL} is equivariant:
\[
 \mathcal{S}(ug,g^{-1}w,g^{-1}w') 
 = \sum\big({\textup{hl}^{\om}}_u(u(ge_i))\vv<e_i,g^{-1}w'>\big)g 
 = \mathcal{S}(u,w,w')g
\] 
for the principal right action of the
structure group, i.e., $g\in {\rm O}(d)$. 
Indeed, this follows since $\textup{hl}^{\omega}: \mathfrak{F}\times_Q
TQ\to\ker\om=\hor^{\om}\subset T\mathfrak{F}$ is ${\rm O}(d)$-equivariant.

To connect with Theorem~\ref{prop:S-equiv}, the principal right action
can be turned to a left action via inversion in the group.

\subsection{Equivariant reduction}\label{sec:equiv-red}
Equivariant reduction is a natural extension of the reduction 
theory of \cite[Theorem~3.1]{LO08a}. While the results of 
\cite{LO08a} are stronger, in the sense that they provide a 
Stratonovich equation on the base space, they are only 
applicable when the original Stratonovich operator is 
$G$-invariant (i.e., equivariant with respect to the trivial 
action on the source space). By contrast, the observation in
equivariant reduction is that although the upstairs Stratonovich
operator is not projectable, the diffusion still factors to a 
diffusion in the base and the downstairs generator is induced 
from that of the original diffusion on the total space. 

Two immediate examples are the construction of Brownian 
motion on a general Riemannian manifold as well as a
stochastic version of Calogero-Moser systems (see below). 
For both of these cases, the diffusion upstairs is defined in 
terms of a Stratonovich operator which is equivariant but not 
projectable.

Let $(\Om,\F,(\F_t),P)$, $Q$, $X_0,X_1,\dots,X_k\in\X(Q)$ and
$\delta\Gamma = \mathcal{S}(Y,\Gamma)\delta Y$ be as before. 
Suppose there is a Lie group $G$ which acts smoothly and properly on $Q$ 
from the left. We continuously extend this action to the one point compactification $\dot{Q}$ 
by requiring $\infty$ to be a fixed point. Let $\pi: Q\toto Q/G$ 
be the projection and $\cinf(Q)^G$ denote the subspace of 
$G$-invariant smooth functions on $Q$.
Note that $Q/G$ need not be a manifold; in general $Q/G$ is a
topological space which is naturally stratified by smooth 
manifolds (see, e.g. \cite[Chapter 2]{DuKo2000}).

In the following all actions are tangent lifted, where appropriate, without further notice.
Generally, for Lie group actions, we will interchangeably use the notation $g\cdot q$
and $gq$.

\begin{theorem}\label{prop:S-equiv}
Given is a group representation $\rho: G\to\textup{O}(k)$ 
and let $\textup{O}(k)$ act on
$\mathbb{R}^{k+1}=\mathbb{R}\times\mathbb{R}^k$ 
such that the first
factor is acted upon trivially.
If the Stratonovich operator
$\mathcal{S}$ satisfies the equivariance property
\begin{equation}\label{e:S-equiv}
 \mathcal{S}(gx,\rho(g)y,\rho(g)y')
 = g\mathcal{S}(x,y,y')
\end{equation}
for all $(x,y,y')\in Q\times T\mathbb{R}^{k+1}$, then the diffusion 
$\Gamma$ induces a diffusion $\pi\circ\Gamma$ in $Q/G$. Moreover, 
the generator $A$ of the diffusion $\Gamma$ on $Q$ preserves 
$\cinf(Q)^G$ and the induced generator $A_0$ of the diffusion 
$\pi\circ\Gamma$ on $Q/G$ is characterized by 
\begin{equation}\label{e:A_0}
 \pi^*(A_0f) = A(\pi^*f)
\end{equation}
for all $f\in\cinf(Q/G) := \set{f\in C(Q/G): \pi^*f\in\cinf(Q)^G}$.
\end{theorem}

\begin{proof}
Let us begin by noting that $g\Gamma^{x,W}=\Gamma^{gx,\rho(g)W}$. Indeed,
\[
 \delta(g\Gamma^{x,W})
 = g\mathcal{S}(Y,\Gamma^{x,W})\delta Y
 = \mathcal{S}(\rho(g)Y,g\Gamma^{x,W})\delta(\rho(g)Y)
\]
whence $\tilde{\Gamma}:=g\Gamma^{x,W}$ satisfies
$\tilde{\Gamma}_0=gx$ a.s.\ and
$\delta\tilde{\Gamma}=\mathcal{S}(\rho(g)Y,\tilde{\Gamma})\delta(\rho(g)Y)$. By
existence and uniqueness of solutions the claim follows. 
In particular, we have 
$\pi\circ\Gamma^{x,W} = \pi\circ\Gamma^{gx,\rho(g)W}$. 

Claim:
\begin{equation}\label{e:claim1}
 P_{gx} = g_*P_x 
\end{equation}
where $G$ acts on $W(Q)$ as $g: w\mapsto(t\mapsto gw(t))$. 
To see this, let $S\subset W(Q)$ be a Borel cylinder set. This means
that there are $l\in\mathbb{N}$, $0\le t_1<\ldots<t_l\in
\mathbb{R}_+$, and a Borel set
$A\subset\Pi^l\dot{Q}$ such that 
$S = \ev(t_1,\dots,t_l)^{-1}(A)$,
where $\ev(t_1,\dots,t_l): W(Q)\to\Pi^l\dot{Q}$,
$w\mapsto(w(t_i))_{i=1}^l$. From the identity
$(\Gamma^{x,\rho(g)W})^{\check{}}_*P = (\Gamma^{x,W})^{\check{}}_*P$
we find
\begin{align*}
 P_{gx}(S)
 &= (\Gamma^{gx,\rho(g)W})^{\check{}}_*P(S)
  = P\set{\om: (\Gamma^{gx,\rho(g)W}_{t_i}(\om))_{i=1}^l \in A}\\
 &= P_x(\ev(t_1,\dots,t_l)^{-1}(g^{-1}A))
  = P_x(g^{-1}S)
\end{align*}
which proves \eqref{e:claim1}.

Consider the push forward map $\pi_*: W(Q)\to W(Q/G)$,
$w\mapsto\pi\circ w$. It is straightforward to see that
$\mathcal{B}(W(Q/G)) = \pi_*\mathcal{B}(W(Q))$. For $S_0 =
\pi_*(S)\in\mathcal{B}(W(Q/G))$ we may write the law 
$\left(P_{[x]}\right)_{[x]\in\dot{Q}/G}$ of $\pi\circ\Gamma$ as
\[
 P_{[x]}(S_0)
 = (\pi\circ\Gamma^{gx,\rho(g)W})^{\check{}}_*P(S_0)
 = P_{gx}(\pi_*^{-1}(S_0)).
\]
By \eqref{e:claim1} this does not depend on $g\in G$. 

Let us show that the system $\left(P_{[x]}\right)_{[x]\in\dot{Q}/G}$ 
satisfies the strong Markov property. Let $p:=\pi_*: W(Q)\to W(Q/G)$, 
$[x]\in Q/G$, $\tau: W(Q/G)\to\mathbb{R}_+$ be a
$(\mathcal{B}_t(W(Q/G)))_t$-stopping time, and $F: W(Q/G)\times
W(Q/G)\to\mathbb{R}$ a bounded
$\mathcal{B}_{\tau}(W(Q/G))\times\mathcal{B}(W(Q/G))$ measurable
function. Then 
$p^{-1}(\mathcal{B}_t(W(Q/G))) \subset \mathcal{B}_t(W(Q))$ 
and
$p^*\tau = \tau\circ p: W(Q)\to\mathbb{R}_+$ is a
$(\mathcal{B}_t(W(Q)))_t$-stopping time.
For $s\in\mathbb{R}_+$ let $\Sigma_s$ be the time shift operator
defined in \eqref{e:time-shift}
and observe that $(\Sigma_spw)(t) = pw(s+t) = \pi(w(s+t)) =
(p\Sigma_sw)(t)$. 
(We use the same notation for the time-shift on $W(Q)$ and that on
$W(Q/G)$.) 
Now, since $P_{[x]}$ is the push forward of $P_x$ via $p$, we can use
the strong Markov property of $(P_x)_x$ to conclude that
\begin{align*}
 &\int_{\{w\in W(Q/G): \tau(w)<\infty\} }F(w,\Sigma_{\tau(w)}w)\,P_{[x]}(dw)\\
 &\qquad \qquad \qquad \qquad\qquad=
 \int_{\{u\in W(Q): p^*\tau(u)<\infty\}}p^*F(u,\Sigma_{p^*\tau(u)}u)\,P_{x}(du)\\
 &\qquad \qquad \qquad \qquad\qquad=
 \int_{\{p^*\tau(u)<\infty\}}
   \Big(
   \int_{W(Q)}p^*F(u,v)\,P_{u(p^*\tau(u))}(dv)
   \Big)P_x(du)\\
 &\qquad \qquad \qquad \qquad\qquad =
 \int_{\{p^*\tau(u)<\infty\}}
   \Big(   
   \int_{W(Q/G)}F(pu,w)\,P_{[u(p^*\tau(u))]}(dw)
   \Big)P_x(du)\\
  &\qquad \qquad \qquad \qquad\qquad =
 \int_{\{\tau(u_0)<\infty\}}
   \Big(   
   \int_{W(Q/G)}F(u_0,w)\,P_{u_0(\tau(u_0))}(dw)
   \Big)P_{[x]}(du_0)
\end{align*}
which, according to \eqref{e:strong-markov}, 
shows that $(P_{[x]})_{[x]}$ is strong Markov. 

To show that $\sum X_iX_i f\in\cinf(Q)^G$ for all $f\in\cinf(Q)^G$ 
consider the standard basis $\{e_0,e_1,\dots,e_k\}$ of 
$\mathbb{R}\times\mathbb{R}^k$. For $j=1,\dots,k$ we find
\[
 g\cdot X_j(x)
 =
 g\cdot \mathcal{S}(x,y,e_j)
 =
 \mathcal{S}(gx,\rho(g)y,\rho(g)e_j)
 =
 \sum_k g_{kj}X_k(gx),
\]
where $g_{kj} := \vv<e_k,\rho(g)e_j>$ is independent of 
$x\in Q$. Since $\sum_j g_{ij}g_{kj} = \delta_{ik}$,
\[
 X_i(gx) 
 = \sum_{j,k} g_{ij}g_{kj}X_k(gx)
 = \sum_j g_{ij}g\cdot X_j(x).
\] 
Thus $\Big(df(X_i)\Big)(gx)=\sum_j g_{ij}\Big(df(X_j)\Big)(x)$ for 
$f\in\cinf(Q)^G$ and also 
\[
 d\Big(df(X_i)\Big)(gx)\circ T_xg
 = d\Big(\sum_j g_{ij}df(X_j)\Big)(x)
 = \sum_j g_{ij}d\Big(df(X_j)\Big)(x).
\]
This implies that 
\begin{align*}
 \sum_i\Big(X_iX_if\Big)(gx)
 &= \sum_i\left\langle d \Big(df(X_i)\Big)(gx), X_i(gx)\right\rangle\\
 &= \sum_{i,j,k} 
    \left\langle 
      g_{ij}d\Big(df(X_j)\Big)(x)\circ (T_xg)^{-1} , g_{ik}(T_xg)\cdot X_k(x)
    \right\rangle\\
 &= \sum_i\Big(X_iX_if\Big)(x).
\end{align*}
Similarly, it is also easy to see that $X_0$ is $G$-invariant. Thus the
generator $A = X_0+\by{1}{2}\sum X_iX_i$ acts on $\cinf(Q)^G$,
whence it induces a projected operator $A_0$ characterized by
$A\circ\pi^* = \pi^*\circ A_0$.

Finally, 
to see that $A_0$ is the generator of $\pi\circ\Gamma$
we need to show that, for all $t\in\mathbb{R}_+$, $[x]\in Q/G$, 
and $f\in\cinf(Q/G)_0$, the $\mathbb{R}$-valued process 
\begin{align*}
 & M_t^f: W(Q/G)\longto\mathbb{R},\\
 & M_t^f(w) 
   := f(w(t))-f(w(0))-\int_0^t(A_0f)(w(s))\,ds
\end{align*}
is a $P_{[x]}$-martingale on $(W(Q/G),\mathcal{B}(W(Q/G))$ for the filtration
$(\mathcal{B}_t(W(Q/G)))_t$. 
See \cite[Def.~IV.5.3]{IW89}.
This means that 
for all $t\ge0$,  $s\in[0,t]$, and $A\in\mathcal{B}_s(W(Q/G))$
we should check that (see \cite[Chapter~V]{Str})
\[
 \int_A E^{P_{[x]}}\Big[M_t^f \Big| \mathcal{B}_s(W(Q/G))\Big](w)\,P_{[x]}(dw)
 = \int_A M_s^f(w)\,P_{[x]}(dw);
\]
$E^{P_{[x]}}$ denotes the expectation on  
$\Big(W(Q/G),\mathcal{B}(W(Q/G))\Big)$ with respect to $P_{[x]}$.
Indeed,
\begin{align*}
 \int_A E^{P_{[x]}}\Big[M_t^f \Big| \mathcal{B}_s(W(Q/G))\Big](w)\,P_{[x]}(dw)
 &=  
 \int_A M_t^f(w)\,P_{[x]}(dw)\\
 &= 
 \int_{p^{-1}A}(p^*M_t^f)(u)\,P_x(du)\\
 &= 
 \int_{p^{-1}A}(\hat{M}_t^{\pi^*f})(u)\,P_x(du)\\
 &=
 \int_{p^{-1}A} E^{P_x}\Big[\hat{M}_t^{\pi^*f} \Big| \mathcal{B}_s(W(Q))\Big](u)\,P_x(du)\\
 &=
 \int_{p^{-1}A} \hat{M}_s^{\pi^*f}(u)\,P_x(du)\\
 &=
 \int_{p^{-1}A}(p^*M_s^f)(u)\,P_x(du)\\
 &=
 \int_A M_s^f(w)\,P_{[x]}(dw).
\end{align*}
Here, $\hat{M}_t^{\pi^*f}: W(Q)\to\mathbb{R}$ is analogously 
defined to $M_t^{f}$. We have used that $\hat{M}_t^{\pi^*f}$ 
is a $P_x$-martingale with respect to $(\mathcal{B}_t(W(Q)))_t$ 
for all $x\in Q$ and that $p^*M_t^f = \hat{M}_t^{\pi^*f}$ 
which holds because of $(A_0f)\circ\pi= A(\pi^*f)$.
\end{proof}

\subsubsection{Stochastic Calogero-Moser systems}
To construct classical trigonometric or rational Calogero-Moser 
models one can take the configurations space $Q$ to be a 
(real or complex) semisimple Lie group $G$ or a semisimple Lie 
algebra $\gu$, respectively. The metric $\mu$ on $Q$ is then 
accordingly given by the (essentially unique)
bi-invariant (pseudo-)metric 
in the group or the $\Ad$-invariant non-degenerate bilinear form 
in the Lie algebra. Thus one obtains a $G$-invariant 
Hamiltonian system $(T^*Q,\Om^Q,\Ham)$ where $\Om^Q$ is the 
canonical symplectic form on $T^*Q$, $\Ham$ is the kinetic 
energy Hamiltonian, and $G$ acts by the cotangent lift of
the conjugation action or the adjoint action, respectively. The
resulting Calogero-Moser system is then realized by passing to 
the (singular) symplectic quotient of  $(T^*Q,\Om^Q,\Ham)$
with respect to the $G$-action. See \cite{KKS78,FP06,H04}.
In other words, Calogero-Moser systems are obtained by reducing 
the Hamiltonian description of 
geodesic motion on the Riemannian manifold $(Q,\mu)$ with 
respect to its obvious symmetry group. 

Here we propose the stochastic analogue of this construction 
which should consist of reducing the Hamiltonian construction of 
Brownian motion on $(Q,\mu)$ with respect to the $G$-action. 
To this end, we consider the Hamiltonian version in \cite{LO08} 
of \eqref{e:SL}.  Using the left trivialization we may write 
$TQ=Q\times\gu$ (recall that $Q = G$ or $Q =\mathfrak{g}$) and 
choose an orthonormal basis $L_i$ of $\gu$ 
with respect to the $\Ad$-invariant inner product 
$\vv<\cdot ,\cdot >$; suppose from now on, for simplicity of exposition,
that $G$ is compact. We obtain a $\gu$-valued Hamiltonian
\[
 H  = (H^i): T^*Q = Q\times\gu^*\longto\gu,
 \quad
 (q,p)\longmapsto\sum\vv<p,L_i>L_i.
\]
The Hamiltonian version of Brownian motion is determined by 
the associated Stratonovich equation
\[
 \delta\Gamma
 = \sum_i X_{H^i}(\Gamma)\vv<L_i,\delta W>,
\]
where $W$ is Brownian motion in $\gu\cong\mathbb{R}^n$. It is 
shown in \cite{LO08} that $\tau\circ\Gamma$ is Brownian motion 
in $(Q,\mu)$, where  $\tau: T^*Q\to Q$ is the projection. 
In the left trivialization $T^*Q = Q\times\gu^*$ the 
Hamiltonian $H$ is nothing but the projection onto the second 
factor when $\gu$ and $\gu^*$ are identified. Clearly, $H$ is 
not $G$-invariant but it is $G$-equivariant for the 
$\Ad$-action on $\gu$. It is easy to see that the same is 
true for the Stratonovich operator
$(q,p;w,w')\mapsto\sum X_{H^i}(q,p)\vv<L_i,w'>$.
In fact, we are ultimately concerned with the Stratonovich 
equation $\delta(\tau\circ\Gamma) = 
\mathcal{S}(W,\Gamma)\delta W 
= \sum\delta W^i L_i = \delta W$ and now it is evident that
$\mathcal{S}(g\cdot q,g\cdot (w,w')) = \Ad(g)w'$ whence we need 
the $\Ad(G)$-action on $(w,w')$ to make the Stratonovich 
operator $\mathcal{S}: Q\times T\gu\to TQ$ equivariant for 
the respective actions. Thus the above theorem applies and 
we obtain a diffusion $\pi\circ\tau\circ\Gamma$ in the 
(singular) space $Q/G$ when $\pi: Q\toto Q/G$ is the projection.

This construction has been carried out in \cite{H11} where it is shown that the associated 
stochastic Hamilton-Jacobi equation of \cite{LO09} is related to the 
quantum Calogero-Moser Schr\"odinger equation of \cite{OP78, OP83}.

The issue of equivariant reduction leads immediately to the 
setting of \cite{ELL04,ELL10}. There, one of the topics treated 
is that of a diffusion on the total space
of a principal bundle such that the diffusion factors through 
the projection and the generator induces also a generator on the 
base.

\subsection{Reconstruction of an equivariant diffusion}
\label{sec:equiv-diff}
Proposition~\ref{prop:mean-re} that we shall state and prove in 
this subsection will be used in the examples considered later 
on. Before we can state it, we need to recall some
notions of \cite{ELL04,ELL10}.

Let $\pi: Q\toto Q/G=:M$ be a left $G$-principal bundle with 
connection form $\A\in\Om^1(Q;\gu)$. Denote the horizontal and 
vertical spaces by $\hor$ and $\ver$, respectively. Assume 
that $\Gamma$ is a diffusion in $Q$ generated by a Stratonovich 
equation 
\begin{equation}\label{e:G_ELL}
 \delta\Gamma
 = X_0(\Gamma)\delta t + \sum_{a=1}^m X_a(\Gamma)\delta W^a
   + v_0(\Gamma)\delta t + \sum_{\alpha=1}^k v_{\alpha}(\Gamma)
   \delta B^{\alpha}
\end{equation}
such that $X_0,X_1,\dots,X_m$ are basic, $v_0,v_1,\dots,v_k$ are
vertical vector fields, and $(W,B)$ is Brownian motion in $\mathbb{R}^{m+k}$ 
with respect to the underlying filtered probability space. 
Thus the generator of $\Gamma$ is 
\[
 A^Q = X_0 + \by{1}{2}\sum X_aX_a + v_0 + 
 \by{1}{2}\sum v_{\alpha}v_{\alpha}.
\] 
By construction, this generator can be decomposed as follows. 
There are $Y_0,Y_1,\dots,Y_m\in\X(M)$ such that
$X_0=\hl(Y_0),\dots,X_m=\hl(Y_m)$ and $x_t := \pi\circ\Gamma_t$ 
is a diffusion in $M$ with generator $A^M =  Y_0 + 
\by{1}{2}\sum Y_aY_a$. Note that $\pi^*\circ A^M = 
A^Q\circ \pi^*$, that is, $A^Q$ is projectable. Moreover, 
$A^Q$ decomposes into a horizontal part $A^h =
X_0 + \by{1}{2}\sum X_aX_a$ and a vertical part $A^v =
v_0(\Gamma)\delta t + \sum_{\alpha=1}^k v_{\alpha}
\delta B^{\alpha}$. 

In \cite{ELL04,ELL10} one of the main points is that,
assuming a non-degeneracy condition, the induced
operator $A^M$ gives rise to a connection in $\pi: Q\to M$ 
with respect to which the operator $A^Q$ can be decomposed. 
In our applications the connection is given by the
problem and the decomposition into horizontal and vertical part 
arises naturally. 

We are going to use the observation of \cite{ELL04,ELL10}
that, for $q\in Q$ and $\Gamma_0=q$
a.s., the diffusion $\Gamma$ can be written as 
\begin{equation}\label{e:re}
 \Gamma_t = g_t^{x^h}\cdot x^h_t.
\end{equation}
Here $x^h_t$ is the diffusion in $Q$ with generator $A^h$ and
$x_0^h=q$ a.s. That is, $x^h_t$ is the horizontal lift of the 
$A^M$ diffusion $x_t$. The process $g_t^{x^h}$ in $G$ with 
$g_0^{x^h}=e$ a.s.\ can be written as the solution to a
time-dependent Stratonovich equation: for $w\in W(Q)$ we 
define
\begin{equation}\label{e:g^w}
 \delta g^w_t 
 = 
 T_eR_{g_t^w}
 \Big(
   \A_{g_t^w\cdot w_t}v_0(g_t^w\cdot w_t)
   + \sum \A_{g_t^w\cdot w_t}v_{\alpha}(g_t^w\cdot x^w_t)
 \Big).
\end{equation}
Here $R_g: G\to G$ is the action by right multiplication of 
$G$ on itself. Equation~\eqref{e:re} is reminiscent of a 
well-known concept in mechanics and can be viewed as a 
reconstruction equation (see, e.g., \cite[\S4.3]{AbMa1978},
\cite[\S3]{MaMoRa1990}, \cite[Theorem~11.8]{Mon}).

Let $Q_e(w)$ be the law on $W(G)$ of \eqref{e:g^w}. This 
depends only on $\pi\circ w\in W(M)$. Let $x_t$ be an 
$A^M$-diffusion path in $M$ with horizontal lift $x^h$. 
Consider the evaluation map $\ev_t: W(G)\to G$, 
$g(\cdot )\mapsto g(t)$. We call
\[
  E^{Q_e(x^h)}[\ev]\cdot x^h: 
  t\longmapsto E^{Q_e(x^h)}[\ev_t]\cdot x_t^h
\]
the \emph{mean reconstruction} of the sample path $x_t$.

From now on, we shall assume that $G$ can be realized as a 
matrix group $G\subset\GL(N)\subset\mathbb{R}^{N^2}$. The 
flat connection on $\mathbb{R}^{N^2}$ thus induces a 
connection $\nabla$ on $G$. For $X\in\gu\subset\gl(N)$ we
denote the associated left- and right-invariant vector field by 
$L_*X (g) = T_eL_g(X) = g X\in\gl(N)$ 
and $R_*X (g) = T_eR_g(X) = Xg\in\gl(N)$.    

\begin{proposition}\label{prop:mean-re}
If $x\in W(M)$ is an $A^M$-sample path 
and $v_0,\dots,v_k$ are $G$-invariant vector fields then
the expectation $E^{Q_e(x^h)}[\ev_t] =: c(t)$ associated to the mean
reconstruction 
of an $A^M$-diffusion path $x$ in $M$
is given as the solution to
the left-invariant time-dependent ODE
\[
 T_eL_{c(t)}^{-1}\left(c'(t)\right)
 = 
 \A_{x^h_t}v_0(x^h_t) 
 + \by{1}{2}\sum_{\alpha=1}^k\nabla_{L_*\big(\A_{x^h_t} v_{\alpha}(x^h_t)\big)}
 L_*(\A_{x^h_t} v_{\alpha}(x^h_t)) (e).
\]
\end{proposition}

\begin{proof}
We can use $G$-invariance of the vector fields together with the
equivariance property $\A_{gq}(gu_q)=\Ad(g)\A_q u_q$, 
$q\in Q$, $u_q\in
T_qQ$, of the principal bundle connection form $\A$ to
rewrite the defining equation \eqref{e:g^w} as 
\[
 \delta g_t^{x^h}
 =
 T_eR_{g_t^{x^h}}\left(T_e(R_{g_t^{x^h}}^{-1}\circ 
 L_{g_t^{x^h}})\left(\A_{x^h_t}
 \Big(v_0(x^h_t)\delta t 
     + \sum v_{\alpha}(x^h_t)\delta B^{\alpha} \Big)
     \right)\right).
\]
Letting
\begin{align*}
 g(t) &= (g(t)^l_n)_{nl} := g_t^{x^h},\\
 a(t)_{\alpha}
 &= (a(t)_{\alpha}^{mn})_{mn} := \A_{x^h_t}v_{\alpha}(x^h_t)\in\gl(N),\\
 b(t) 
 &= (b(t)^{mn})_{mn}
 := \A_{x^h_t}v_0(x^h_t)\in\gl(N),
\end{align*}
this becomes with the summation convention, for $l=1,\dots,N$, the  Stratonovich
equation 
\[
 \delta g^l
 = 
 \Big(g^l_m a_{\alpha}^{mn}\delta B^{\alpha}  
      + g^l_m b^{mn}\delta t
 \Big)_{n=1}^N
\]
in $\mathbb{R}^N$ when we think of $g^l$ as a column vector
and suppress the time-dependency.
The associated It\^o equation in $\mathbb{R}^N$ is, for $l=1,\dots,N$,
\[
 d g^l
 = 
 \Big(
   g^l_m a_{\alpha}^{mn} dB^{\alpha}  
   + ( g^l_m b^{mn} 
   +   \by{1}{2}a_{\alpha}^{rn} g^l_m a_{\alpha}^{mr}) dt
 \Big)_{n=1}^N.
\]
(See e.g.\ \cite[Equ.~(6.1.3)]{Oks07} for the conversion rule of Stratonovich equations
to It\^o equations.)
This is a linear time-dependent It\^o equation in $\mathbb{R}^N$. Hence, the mean
motion is found by erasing the martingale term in the corresponding
integral equation. This implies that the expected motion of $g$ is
given by
\[
 c'(t)
 = \by{d}{dt}E[g](t)
 = E[g](t)\left(b(t) + \by{1}{2}a(t)_{\alpha}a(t)_{\alpha}\right)
\]
which is an equation in $\GL(N)$. Since $a(t)_{\alpha}a(t)_{\alpha} =
\nabla_{L_*a(t)_{\alpha}} (L_*a(t)_{\alpha})(e)$ the claim follows. 
\end{proof}

\subsection{Time reversible diffusions}\label{sec:trd}

The references for this section are \cite{Kol36,Ken78,IW89}.
Let $(M,\mu)$ be a Riemannian manifold and $\Gamma$ an 
$A$-diffusion in $M$ where 
\begin{equation}\label{e:trdA}
 A = \by{1}{2}\Delta + \by{1}{2}b
\end{equation}
with $\Delta$ the Laplace-Beltrami operator and $b$ a vector 
field. Let $p(t,x,y)$ denote the transition probability 
density of $\Gamma$ (the minimal fundamental solution -- 
see \cite{Ken78}). If $\vol_{\mu}$ is the Riemannian volume 
form on $M$ then, for $(t,x,S)\in\mathbb{R}_+\times M\times
\mathcal{B}(M)$, the transition probability of $\Gamma$ is 
\[
 P(t,x,S) = \int_S p(t,x,y)\vol_{\mu}(y).
\]
This quantifies the probability that a diffusion path starting 
at $x$ is in $S$ after time $t$. The diffusion $\Gamma$ is 
said to be \textit{symmetrizable} if there is a smooth function 
$\phi>0$ such that 
\begin{equation}\label{e:symm}
 p(t,x,y)\phi(x) = p(t,y,x)\phi(y)\;
 \textup{ for all } t,x,y\in\mathbb{R}_+\times M\times M
\end{equation}
in which case $\Gamma$ is called $\phi$-\textit{symmetric}.

A probability measure $\nu$ on $M$ is an 
\textit{equilibrium measure} if $\int_M\nu = 1$ and
\[
P(t,x,S)\longto\nu(S) 
\textup{ as }
t\longto\infty 
\]
for all 
$(x,S)\in M\times\mathcal{B}(M)$. 
Equilibrium measures, if they exist, are unique.
If $\nu = \phi\vol_{\mu}$ is an equilibrium measure 
then we refer to $\phi$ as the
\textit{equilibrium distribution}.

The diffusion  $\Gamma$ is called \textit{time-reversible} 
if its law coincides with that of the time-reversed
process; this means that for each $T>0$ the law $P_{[0,T]}$ of $[0,T]\times\Om\to M$, 
$(t,\om)\mapsto\Gamma_t(\om)$ is the
same as the law $P_{[0,T]}^-$ of $[0,T]\times\Om\to M$,
$(t,\om)\mapsto\Gamma_{T-t}(\om)$. 

The adjoint operator $A^*$ associated to $A$ is given by 
\[
 A^*f = \by{1}{2}\Delta f  - \by{1}{2}\textup{div}_{\mu}(fb)
\]
where $f\in\cinf(M)$.  
Here, the adjoint is with respect to the $L^2$ inner product
$\vv<f,g> = \int_M fg\,\vol_{\mu}$.
The following result is essentially due to Kolmogorov.

\begin{theorem}[\cite{Kol36,Ken78,IW89}]\label{thm:Kol}
With notation as above the following are true. 
\begin{enumerate}[\up (1)]
\item
The $A$-diffusion $\Gamma$ is
symmetric if and only if $b$ is a gradient. Moreover, if 
$b=\textup{grad}(\log\phi)$ then $\Gamma$ is $\phi$-symmetric 
and $A^*\phi=0$. 
\item
$\Gamma$ is time reversible if and only if it is symmetric and 
has an equilibrium distribution $\phi$, in which case 
$\Gamma$ is $\phi$-symmetric.
\item
If $M$ is compact then an equilibrium distribution always exists.  
\item
If $M$ is compact 
then the unique equilibrium distribution $\phi$ is  characterized by
the equations 
$\int_M\phi\vol_{\mu}=1$ and $A^*\phi=0$.
\end{enumerate}
\end{theorem}

Compactness of $M$ is satisfied in important examples such the
Chaplygin ball or the two-wheeled carriage studied in 
Section~\ref{sec:5}.

Assuming that $M$ is compact, \cite[Chapter~5]{JQQ04} give 
various equivalent conditions for a diffusion of the 
form \eqref{e:trdA} to be time-reversible. One such condition 
is that the diffusion have vanishing entropy production rate 
\begin{equation}\label{epr}
 \lim_{T\to0}\by{1}{T}H(P_{[0,T]},P^-_{[0,T]}).
\end{equation}
The \textit{relative entropy} $H(\mu,\nu)$ of two probability measures 
$\mu,\nu$ on a measure space $(\mathcal{W},\mathcal{B})$ is defined as (see \cite[Definition~1.4.3]{JQQ04})
\[
 H(\mu,\nu)
 :=
 \left\{
 \begin{matrix}
  \int_{\Om}\log\by{d\mu}{d\nu}\mu(d\om)
  && \textup{if }\mu<<\nu \textup{ and }\log\by{d\mu}{d\nu}\in L^1(d\mu);\\
  +\infty && \textup{otherwise}.
 \end{matrix}
\right.
\]

\section{Non-holonomic diffusions}\label{sec:4}

Consider a non-holonomic system $(Q,\mathcal{D},\mathcal{L})$ as in
Section~\ref{sec:2}.
This section is concerned with the study of non-holonomic diffusions on
$\mathcal{D}$ which should be given by a Stratonovich equation of the form
\begin{equation}\label{e:nh-diff}
 \delta\Gamma 
 = 
 X^{\mathcal{C}}_{\mathcal{H}} \delta t 
  + \mathcal{S}^{\mathcal{C}}(\Gamma,W)\delta W.
\end{equation}
Here $X^{\mathcal{C}}_{\mathcal{H}}$ describes the
dynamics of the deterministic system,
$W$ is Brownian motion in $\mathbb{R}^d$, $d=\dim Q$,  
and 
$\mathcal{S}^{\mathcal{C}}(\Gamma,W)\delta W$ should be interpreted as a
noise term that stems from
\emph{constrained Brownian motion}.
This is in analogy to \cite[Section~3.1]{LO08} and \cite{Bis81} where
Hamiltonian diffusions are introduced. 
However, equation \eqref{e:nh-diff} does not make sense, in general,
unless the configuration space is parallelizable.
The problem to be considered below is to make this equation precise 
and to study the notion of constrained Brownian motion on manifolds.

\subsection{Constrained Brownian motion}\label{sec:CBM}
Let $(Q,\mathcal{D},\mathcal{L})$ be a non-holonomic system with symmetry group $G$
as in Section~\ref{sec:2}.

Let $\rho: \mathfrak{F}\to Q$ be the orthonormal frame bundle over
$(Q,\mu)$ and denote its structure group by $K:=\textup{O}(d)$ and
its Lie algebra by $\mathfrak{k}: = \mathfrak{so}(d)$. 
The Levi-Civita connection $\nabla^{\mu}$ on $TQ$ gives rise to a
uniquely determined principal connection $\omega\in 
\Omega^1\left(\mathfrak{F};\mathfrak{k}\right)$ on the principal
bundle $\rho: \mathfrak{F}\to Q$; denote by $\hor^{\omega} = 
\ker\omega\subset T\mathfrak{F}$ its horizontal subbundle.
Equip $\mathfrak{F}$ with a $K$-invariant metric $\nu$ such that 
$\rho$ becomes a Riemannian submersion and $\hor^{\omega}$ and 
$\ver(\rho) := \ker T\rho$ are
perpendicular. Since $G$ acts by isometries on $(Q,\mu)$, it lifts to
an action on $\mathfrak{F}$ and we may assume $\nu$ to be 
$G$-invariant; e.g., we could take $\nu$ to be the Sasaki-Mok metric 
(see the survey \cite{Sek08}). 
We can use the connection to lift the constraints to a subbundle
$\mathcal{D}\mathfrak{F}\subset T\mathfrak{F}$ defined via the 
natural $\omega$-dependent vector bundle isomorphism
\[
 \mathcal{D}\mathfrak{F} 
 \cong_{\om} (\mathfrak{F}\times_Q\mathcal{D})\oplus\ver(\rho).
\]
To understand this definition and the 
isomorphism consider the bundle morphism over $\mathfrak{F}$
defined by
\begin{align*}
 \mathfrak{F}\times_Q\mathcal{D}\oplus\ver(\rho)
 &\hookto
 \hor^{\om}\oplus\ver(\rho)
 \cong_{\om}T\mathfrak{F},\\
 (u_q,X_q;\eta_{u_q})
 &\longmapsto
 \textup{hl}^{\om}_{u_q}(X_q)+\eta(u_q).
\end{align*}
As before, $\textup{hl}^{\om}: \mathfrak{F}\times_Q TQ\to\hor^{\om}$ 
is the horizontal lift mapping associated to $\om$.
Now the subbundle $\mathcal{D}\mathfrak{F}$ is defined as the image of this morphism.

Thus $(\mathfrak{F},\mathcal{D}\mathfrak{F},\by{1}{2}||\cdot||_{\nu})$ 
is a new $G$-invariant non-holonomic system covering 
$(Q,\mathcal{D},\mathcal{L})$ in the following sense: 
The $G$-action lifts to an action on $\mathcal{D}\mathfrak{F}$
and there is an induced space $\C\mathfrak{F}$ 
defined by
\[
 \C\mathfrak{F}: = \set{\xi\in T(\mathcal{D}\mathfrak{F})\mid 
  T\tau_{\mathfrak{F}}(\xi)\in\mathcal{D}\mathfrak{F}},
\]
where $\tau_{\mathfrak{F}}: T\mathfrak{F}\to\mathfrak{F}$ is the 
tangent bundle projection. Again, we split 
$T(T\mathfrak{F})|(\mathcal{D}\mathfrak{F}) = \C\mathfrak{F}\oplus(\C\mathfrak{F})^{\Om_{\mathfrak{F}}}$,
where $\Omega_{\mathfrak{F}}$ is now the canonical symplectic
form on $T\mathfrak{F}\cong_{\nu}T^*\mathfrak{F}$ (the tangent and
cotangent bundles of $\mathfrak{F}$ are identified via the Riemannian
metric $\nu$ on $Q$), and $P_{\mathfrak{F}}: \C\mathfrak{F}\oplus
(\C\mathfrak{F})^{\Omega_{\mathfrak{F}}}\to\C\mathfrak{F}$ denotes the
associated projection. The situation is summarized in the following 
commutative diagram:
\begin{equation}\label{e:diag}
\xymatrix{
 {T(T\mathfrak{F})|(\mathcal{D}\mathfrak{F})}
  \ar @{=}[r]
  \ar @{->}[d]_{TT\rho}
 &
 {\C\mathfrak{F}\oplus(\C\mathfrak{F})^{\Omega_{\mathfrak{F}}}}
  \ar @{->}[r]^-{P_{\mathfrak{F}}}
  \ar[d]
 & 
 {\C\mathfrak{F}}
  \ar[d]\\
 {TTQ|\mathcal{D}}\ar @{=}[r]
 &
 {\C\oplus\C^{\Omega}}\ar @{->}[r]^-{P}
 &
 {\C}.
}
\end{equation}
Indeed, this diagram is commutative since $TT\rho(\C\mathfrak{F})=\C$ and 
we may regard $(TQ,\Omega)$ as the
symplectic reduction of $(T\mathfrak{F},\Omega_{\mathfrak{F}})$ with 
respect to the $K$-action at $0$. 
In particular, 
\[
TT\rho \left(P_{\mathfrak{F}}\left(X_{\rho^*f}\right)\right) 
= P\left(TT\rho\left(X_{\rho^*f}\right)\right) 
= X^{\mathcal{C}}_f\in\X(\mathcal{D})
\] 
for any $f\in\cinf(TQ)$.

According to Section~\ref{sec:3} we can construct Brownian motion on
$(Q,\mu)$ by fixing Brownian motion $W=(W_i)$ in $\mathbb{R}^d$, $d=\dim Q$,
and the Hamiltonian
\begin{equation}\label{e:H}
 H: T\mathfrak{F}\longto\mathbb{R}^d,\quad
 (u,\eta)\longmapsto\big(\nu(\eta,L_i(u))\big) = \big(H^i(u,\eta)\big),
\end{equation}
where the $L_i$ are defined by \eqref{e:L}.
This gives rise to the Stratonovich operator 
\[
 \mathcal{S}^H: 
 T\mathfrak{F}\times T\mathbb{R}^d\longto TT\mathfrak{F},\quad
 (u,\eta,x,w)\longmapsto\sum X_{H^i}(u,\eta)\vv<e_i,w>,
\]
where $\{e_1, \ldots, e_d\}$ is the standard basis in $\mathbb{R}^d$. 
If $\Gamma^H$ solves $\delta\Gamma^H 
= \mathcal{S}^H(W,\Gamma^H)\delta W$ then
$\tau_{\mathfrak{F}}\circ\Gamma^H$ solves \eqref{e:SL} and
$\rho\circ\tau_{\mathfrak{F}}\circ\Gamma^H$ is Brownian motion in
$(Q,\mu)$.

\begin{definition}
We define \emph{constrained Brownian motion} to be the
process 
\[
 \Gamma^{\textup{nh}} := 
 \rho\circ\tau_{\mathfrak{F}}\circ\Gamma^{\mathcal{C\mathfrak{F}}}
\]
in $Q$, where $\Gamma^{\mathcal{C}\mathfrak{F}}$ is a process in 
$\mathcal{D}\mathfrak{F}$ solving the Stratonovich
equation
\begin{equation}
 \delta\Gamma^{\mathcal{C}\mathfrak{F}}
 =
 P_{\mathfrak{F}}(\Gamma^{\mathcal{C}\mathfrak{F}})\mathcal{S}^H(W,\Gamma^{\mathcal{C}\mathfrak{F}})\delta W
 = \sum X^{\mathcal{C}\mathfrak{F}}_{H^i}(
 \Gamma^{\mathcal{C}\mathfrak{F}}) \delta W^i.
\end{equation}
\end{definition}

Let $(\mathcal{D}\mathfrak{F})^{\bot}$ be the $\nu$-orthogonal of 
$\mathcal{D}\mathfrak{F}$ and
$\Pi_{\mathfrak{F}}: T\mathfrak{F} = \mathcal{D}\mathfrak{F}\oplus(\mathcal{D}\mathfrak{F})^{\bot} \to\mathcal{D}\mathfrak{F}$ be the
orthogonal projection.
Similarly, we define $\Pi: TQ=\mathcal{D}\oplus\mathcal{D}^{\bot}\to\mathcal{D}$
and we note that 
\begin{equation}\label{e:Pi-Pi}
 \Pi\circ T\rho = T\rho\circ\Pi_{\mathfrak{F}}.
\end{equation}
Equation \eqref{e:Pi} implies that
$\tau_{\mathfrak{F}}\circ\Gamma^{\mathcal{C}\mathfrak{F}}$ is a
diffusion in $\mathfrak{F}$ generated by the Stratonovich equation
\[
 \delta(\tau_{\mathfrak{F}}\circ\Gamma^{\mathcal{C}\mathfrak{F}})
 =
 \sum\Pi_{\mathfrak{F}}(\tau_{\mathfrak{F}}\circ
 \Gamma^{\mathcal{C}\mathfrak{F}})
 L_i(\tau_{\mathfrak{F}}\circ\Gamma^{\mathcal{C}\mathfrak{F}})\delta W^i. 
\]
The associated Stratonovich operator will be denoted by
$\mathcal{S}^{\mathcal{D}\mathfrak{F}}$. It is explicitly given by 
\begin{equation}\label{e:SDF}
 \mathcal{S}^{\mathcal{D}\mathfrak{F}}:
 \mathfrak{F}\times T\mathbb{R}^d\longto\mathcal{D}\mathfrak{F},\quad
 (u;w,w')\longmapsto\sum\Pi_{\mathfrak{F}}(u)L_i(u)\vv<e_i,w'>
\end{equation}
where $\{e_1, \ldots, e_d\}$ is the standard basis in $\mathbb{R}^d$. 

Henceforth, local orthonormal frames on $(Q,\mu)$ and local
sections of $\rho: \mathfrak{F}\to Q$ will be identified.

\begin{theorem}\label{prop:Gamma-nh}
The process $\Gamma^{\textup{nh}}$ is a diffusion in $Q$ (in the sense
of Section~\ref{sec:3}) and its generator $A$ has the form 
\[
 A = \by{1}{2}\sum(\Pi u_i)(\Pi u_i)  
      - \by{1}{2}\sum\Pi\nabla^{\mu}_{\Pi u_i}u_i
\]
in a local orthonormal frame $u=(u_i)$ on $(Q,\mu)$.
\end{theorem}

\begin{proof}
By Proposition~\ref{prop:S-equiv}, 
to see that $\Gamma^{\textup{nh}}$ is a diffusion we need to show
that 
\begin{equation}\label{e:equ}
 \mathcal{S}^{\mathcal{D}\mathfrak{F}}(u\cdot k,w,w') 
 = \mathcal{S}^{\mathcal{D}\mathfrak{F}}(u,kw,kw')\cdot k
\end{equation}
for all $k\in K$. Here, $u\cdot k$ denotes the principal right 
action of $k\in K$ on $u\in\mathfrak{F}$ 
and condition \eqref{e:equ} is equivalent to
\eqref{e:S-equiv} since one can invert a right action to 
obtain a left action. Indeed, $\mathcal{D}\mathfrak{F} \cong_{\omega} 
(\mathfrak{F}\times_Q\mathcal{D})\oplus\ver(\rho)$ and the
definition of $\nu$ imply 
$\Pi_{\mathfrak{F}}\circ\textup{hl}^{\om} = \textup{hl}^{\om}\circ\Pi$
and therefore
\begin{align*}
 \mathcal{S}^{\mathcal{D}\mathfrak{F}}(u\cdot k,w,w')
 &=
 \sum_{i=1}^d\Pi_{\mathfrak{F}}(u\cdot k)\textup{hl}^{\omega}_{u\cdot k}(u(ke_i))\vv<e_i,w'>\\
 &=
\left(\sum_{i=1}^d\textup{hl}^{\om}_{u}(\Pi(\rho(u))u(e_i))\vv<e_i,kw'>\right)\cdot k
\end{align*}
which proves \eqref{e:equ}.

In order to calculate the generator $A$, let $f\in\cinf(Q)_0$ and 
$u\in\mathfrak{F}$. Then
\begin{align*}
 A(\rho^*f)(u)
 &= \by{1}{2}\sum_{i=1}^d(\Pi_{\mathfrak{F}}L_i)
 (\Pi_{\mathfrak{F}}L_i)(\rho^*f)(u)\\
 &= \by{1}{2}\sum_{i=1}^d(\Pi u_i,-\om(\Pi u_i))(\Pi u_i,-\om(\Pi u_i))\rho^*f\\
 &= \by{1}{2}\sum_{i=1}^d(\Pi u_i,-\om(\Pi u_i))(\Pi u_if)\\
 &= \by{1}{2}\sum_{i=1}^d \Pi u_i (x\mapsto(\Pi u_i f)(x))
    - \by{1}{2}\sum_{i=1}^d\om(\Pi u_i)(v\mapsto (\Pi v(e_i)f))|_{v=u}\\
 &= \by{1}{2}\sum_{i=1}^d(\Pi u_i)(\Pi u_i)f
    - \by{1}{2}\sum_{i=1}^d\left(\Pi\left.\frac{\partial}{\partial t}\right|_{t=0}\left(e^{t\om(\Pi u_i)}u\right)e_i\right)f\\
 &= \by{1}{2}\sum_{i=1}^d(\Pi u_i)(\Pi u_i)f
    - \by{1}{2}\sum_{i=1}^d(\Pi \om(\Pi u_i)u_i)f\\
 &=  \by{1}{2}\sum_{i=1}^d(\Pi u_i)(\Pi u_i)f
    - \by{1}{2}\sum_{i=1}^d(\Pi \nabla^{\mu}_{\Pi u_i}u_i)f,
\end{align*}
where we have frequently dropped the base points to simplify the
notation.

Alternatively, the above calculation can be done in local coordinates
$u = (x^i,e^i_j)$ on $\mathfrak{F}$. Then, using the summation 
convention, $u_r =
e^m_r\del_m$ and $\Pi\del_m = \Pi_m^n\del_n$ with $\del_m =
\dd{x^m}{}$. The Christoffel symbols are given as usual by
$\nabla^{\mu}_{\del_k}\del_j = \Gamma^i_{kj}\del_i$. 
Using now the local coordinate description \eqref{e:hor-om-coord} of
$\hor^{\om}$ it follows that
\[
 \Pi_{\mathfrak{F}}L_r(u)
 =
 \textup{hl}^{\om}(\Pi u_r)
 =
 \textup{hl}^{\om}\Pi_m^n e^m_r\del_n
 =
 \Pi_m^n e^m_r
 \Big(
  \del_n - \Gamma^i_{nl}e^l_j\dd{e^i_j}{}
 \Big).
\]
This yields
\begin{align*}
 (\Pi_{\mathfrak{F}}L_r)(\Pi_{\mathfrak{F}}L_r)f
 &=
 (\Pi u_r)(\Pi u_r)f
 -
 \Pi^n_me^m_r\Gamma^i_{nl}e^l_j\dd{e^i_j}{}(\Pi^b_ae^a_r\del_b f)\\
 &=
 (\Pi u_r)(\Pi u_r)f - \Pi^n_me^m_r\Gamma^a_{nl}e^l_r\Pi^b_a\del_b f\\
 &=
 (\Pi u_r)(\Pi u_r)f - \Pi\nabla^{\mu}_{\Pi_m^ne^m_r\del_n}e^l_r\del_l f
 =  
 (\Pi u_r)(\Pi u_r)f - \Pi\nabla^{\mu}_{\Pi u_r}u_r f
\end{align*}
which immediately yields the formula for $A$ in the statement of
the theorem.
\end{proof}

The second term in the above formula for $A$ is reminiscent of the
non-holonomic connection $\nabla^{\textup{nh}}$. This is defined as
the linear connection
\begin{equation}\label{e:nabla-nh}
 \nabla^{\textup{nh}}: \mathfrak{X}(Q)\times \mathfrak{X}(Q)\longto \mathfrak{X}(Q),\quad
 (X,Y)\longmapsto
  \nabla^{\mu}_XY - (\nabla^{\mu}_X\Pi)Y.
\end{equation}
If $Y$ is a section of $\mathcal{D}$, one obtains the useful identity 
$\nabla^{\textup{nh}}_XY = \Pi\nabla^{\mu}_XY$. Let
$\textup{Hess}^{\textup{nh}}$ be the Hessian of
$\nabla^{\textup{nh}}$.

\begin{corollary}\label{cor:nh-diff}
The non-holonomic diffusion $\Gamma^{\textup{nh}}$ is a martingale in
$Q$ with respect to the non-holonomic connection
$\nabla^{\textup{nh}}$. 
\end{corollary}

\begin{proof}
Since $A(\rho^*f)(u)$ depends only on $\rho(u)$ 
(this follows from Theorem~\ref{prop:S-equiv} 
but can also be checked directly),
we may take a local frame
$u=(u_i)=(u_a,u_{\alpha})$ which is adapted to the decomposition
$\mathcal{D}\oplus\mathcal{D}^{\bot}$ such that $u_a$ are local
sections of $\mathcal{D}$ and $u_{\alpha}$ are local
sections of $\mathcal{D}^{\bot}$. Then 
$A = \by{1}{2}\sum (u_a u_a - \nabla^{\textup{nh}}_{u_a}u_a) 
= \by{1}{2}\sum \textup{Hess}^{\textup{nh}}(\cdot)(u_a,u_a)$
which is purely second order, by definition.
\end{proof}

\subsection{$G$-Chaplygin diffusions and stochastic non-holonomic reduction}\label{sec:G-Chap-diff}
Continue to assume that the non-holonomic system
$(Q,\mathcal{D},\mathcal{L})$ is invariant with respect to a
free and proper action of a Lie group $G$.
Denote the projection
by $\pi: Q\toto Q/G = M$. 

The $G$- and $K$-action on $\mathfrak{F}$ commute. Thus, we may form the
product action of $G\times K$ on $\mathfrak{F}$. Since the $H^i$ from 
\eqref{e:H} are $G$-invariant, it follows that the Stratonovich operator
\eqref{e:SDF} satisfies condition \eqref{e:S-equiv} with respect to
the trivial $G$-representation on $\mathbb{R}^d$. Therefore,
$\Gamma^{\mathcal{C}\mathfrak{F}}$ induces a diffusion 
\[
 \pi\circ\rho\circ\tau_{\mathfrak{F}}\circ\Gamma^{\mathcal{C}\mathfrak{F}}
 = \pi\circ\Gamma^{\textup{nh}} 
 =: \Gamma^M
\]
on $M := \mathfrak{F}/(G\times K)$.

Now we make the additional assumption that the constraints are of
Chaplygin type, i.e., $\mathcal{D}$ is the kernel of a principal 
connection one-form $\mathcal{A} \in \Omega^1(Q; \mathfrak{g})$.
The non-holonomic connection \eqref{e:nabla-nh} on $Q$ induces
a connection on $M$ which will be referred to as the non-holonomic
connection $\nabla^{M}$ on $M$; it is given by 
\[
 \nabla^M: \mathfrak{X}(M)\times \mathfrak{X}(M)\longto \mathfrak{X}(M),
 \quad 
 (X,Y)\longmapsto T\pi\left(\Pi\nabla^{\mu}_{\textup{hl}^{\mathcal{A}} X}(\hl Y) \right)
\]
where $\hl: \X(M)\to\X(Q,\mathcal{D})$ (the space of vector fields on
$Q$ with values in the vector subbundle $\mathcal{D} \subset TQ$) is the 
horizontal lift map of $\mathcal{A}$. Recall from \S\ref{sec_Chaplygin} 
that the Riemannian metric $\mu$ on $Q$ naturally induces a Riemannian 
metric $\mu_0$ on the quotient $M:=Q/G$. To calculate the generator $A^M$ 
of $\Gamma^M$, take a local orthonormal frame $u=(u_a)$ on $M$. 
Similarly as in the proof of Corollary~\ref{cor:nh-diff}, the generator 
becomes
\begin{equation}\label{e:AM}
 A^M 
 = \by{1}{2}\sum(u_au_a - \nabla^{\mu_0}_{u_a}u_a)
    + \by{1}{2}\sum(\nabla^{\mu_0}_{u_a}u_a - \nabla^{M}_{u_a}u_a)
 = \by{1}{2}\Delta^{\mu_0} + \by{1}{2}b
\end{equation}
where $b = \sum(\nabla^{\mu_0}_{u_a}u_a - \nabla^{M}_{u_a}u_a)$. 

\begin{lemma}\label{lem:b}
$b = \mu_0^{-1}\beta$ where $\beta$ is defined by \eqref{e:beta1}.
\end{lemma}

\begin{proof}
Essentially this formula is a special case of
\cite[Proposition~8.5]{K92}. 
For convenience we provide a proof 
by using a local 
orthonormal frame $(u_a)$ on $M$.
Let $K = \zeta\circ\curv^{\mathcal{A}}\in\Om^2(Q,TQ)$ 
be the curvature of $\zeta\circ\A\in\Om^1(Q,TQ)$
where $\zeta: \mathfrak{g}\ni\xi\mapsto \xi_Q \in \mathfrak{X}(Q)$ 
is the fundamental vector field mapping of the $G$-action. 
Then 
\begin{align*}
 \mu_0(\nabla^{\mu_0}_{u_a}u_a,u_b)
 &= -\mu_0([u_a,u_b],u_a)
  = -\mu(\hl[u_a,u_b],\hl u_a)\\
 &= \mu(K(\textup{hl}^{\mathcal{A}}u_a,\textup{hl}^{\mathcal{A}}u_b),\textup{hl}^{\mathcal{A}} u_a)
    - \mu([\textup{hl}^{\mathcal{A}}u_a,\textup{hl}^{\mathcal{A}}u_b],\textup{hl}^{\mathcal{A}}u_a)\\
 &= -\mu(\zeta_{\textup{Curv}^{\mathcal{A}}(\textup{hl}^{\mathcal{A}}u_a,\textup{hl}^{\mathcal{A}}u_b)}, \textup{hl}^{\mathcal{A}}u_a)
    + \mu(\Pi\nabla^{\mu}_{\textup{hl}^{\mathcal{A}}u_a}\textup{hl}^{\mathcal{A}}u_a, \textup{hl}^{\mathcal{A}}u_b)\\
 &= -\vv<J(\textup{hl}^{\mathcal{A}}u_a),\curv^{\mathcal{A}}(\textup{hl}^{\mathcal{A}}u_a,\textup{hl}^{\mathcal{A}}u_b)>
    + \mu_0(\nabla^M_{u_a}u_a,u_b).
\end{align*}
Therefore, 
\begin{align*}
 \check{\mu}_0(\nabla^{\mu_0}_{u_a}u_a - \nabla^M_{u_a}u_a)
 &=
 \vv<J(\textup{hl}^{\mathcal{A}}u_a),\curv^{\mathcal{A}}(\textup{hl}^{\mathcal{A}}u_b,\textup{hl}^{\mathcal{A}}u_a)>
 \mu_0(u_b,\_)\\
 &=
 \Xi(u_a)\Big(X^{\textup{nh}}_{\mathcal{H}_{\textup{c}}},u_a^h\Big) 
 = \beta
\end{align*}
where $u_a^h$ is the horizontal lift of the local vector field 
$u_a\in\X_{\textup{loc}}(M)$ to $u_a^h\in\X_{\textup{loc}}(TM)$ with
respect to the Levi-Civita connection $\nabla^{\mu_0}$. 
(We have identified linear functions on $TM$ and one-forms on $M$ as
we did in the definition \eqref{e:beta1}.) 
\end{proof}

\begin{theorem}\label{thm:chap-diff}
The $G$-Chaplygin system $(Q,\mathcal{D},\mathcal{L})$ has a preserved measure if and
only if the associated diffusion $\Gamma^M$ is symmetric. 
Moreover,
if $b=\textup{grad}^{\mu_0}(\log\N)$ then the diffusion is
$\N$-symmetric and $\N$ is the density of the preserved measure of
$\Xnh$ with respect to the Liouville volume.
\end{theorem}

\begin{proof}
Using \eqref{e:AM} and Lemma~\ref{lem:b} this is a direct consequence of
Proposition~\ref{prop:pm} and Theorem~\ref{thm:Kol}.
\end{proof}

When $M$ is compact, then we infer from Section~\ref{sec:trd} that 
measure preservation of the deterministic system is equivalent to
time-reversibility of $\Gamma^M$ which in turn is equivalent to the
vanishing of the entropy production rate \eqref{epr} of $\Gamma^M$. 
Moreover,
if $b=\textup{grad}^{\mu_0}(\log\N)$ then
\[
\left(\int_M\N\vol_{\mu_0}\right)^{-1}\N\vol_{\mu_0}  
\]
is the (unique)
equilibrium distribution of $\Gamma^M$.
For most systems of practical interest, such as 
the Chaplygin ball, the two-wheeled robot, or the snakeboard, the
manifold $M$ is compact.

\section{Examples}\label{sec:5} 

\subsection{The two-wheeled robot}
The configuration space of the two-wheeled robot is 
\[
 Q = S^1\times
S^1\times \textup{SE}(2) = \set{(\psi^1,\psi^2,x,y,\theta)}. 
\]
Here
$(\psi^1,\psi^2)$ measure the positions of the wheels with the
orientation such that the robot goes forward when the wheels go
backward, and $(x,y,\theta)$ give the overall configuration of the
robot in the plane. 
Let $G = \textup{SE}(2)$ and $M := S^1\times S^1 = Q/G$.
We use almost exactly the same notation as
\cite[Section~5.2.2]{M02}. It is assumed that the two wheels can be
controlled independently and roll without slipping and without 
lateral sliding on the plane. 
The Lagrangian $\mathcal{L}$ of the system is the kinetic energy 
corresponding to the metric
\begin{align*}
 \mu
 &= J_w(d\psi^1\otimes d\psi^1 + d\psi^2\otimes d\psi^2)
   + m(dx\otimes dx + dy\otimes dy)\\
   &\phantom{==} 
   + m_0l\cos\theta(dy\otimes d\theta + d\theta\otimes dy)
   - m_0l\sin\theta(dx\otimes d\theta + d\theta\otimes dx)
   + J_0d\theta\otimes d\theta.
\end{align*}
Here $m = m_0+2m_w$, $m_0$ is the mass of the robot without the
wheels, $m_w$ is the mass of each wheel, $J_w$ is the moment of
inertia of each wheel, $J_0$ is the moment of inertia of the robot about 
the vertical axis, and $l$ is
the distance from the vehicle's center of mass to the midpoint of the
axis which connects the two wheels. 
Let $2c$ denote the distance between the  contact points of the two 
wheels with the ground, and $R$ the radius of the wheels. 
The constraints are given by the kernel of the $\gu\cong\mathbb{R}^3$-valued one-form 
\[
 \A 
 =
 \left(
 \begin{matrix}
 \vspace{1mm}
   dx + yd\theta + \by{R}{2}\cos\theta(d\psi^1 + d\psi^2) 
                 +  y\by{R}{2c}(d\psi^1 - d\psi^2)\\
                 \vspace{1mm}
   dy - xd\theta + \by{R}{2}\sin\theta(d\psi^1 + d\psi^2) 
                 -  x\by{R}{2c}(d\psi^1 - d\psi^2)\\
   d\theta + \by{R}{2c}(d\psi^1 - d\psi^2)
 \end{matrix}
 \right).
\]
Thus the constraint distribution is $\mathcal{D} = \A^{-1}(0) =
\textup{span}\set{\xi_1,\xi_2}$ 
where 
\begin{equation}\label{e:xi_i}
 \xi_1
 :=
 \del_{\psi^1} - \by{R}{2}(\cos\theta\del_x + \sin\theta\del_y + \by{1}{c}\del_{\theta})\quad 
 \textup{and} \quad 
 \xi_2
 :=
 \del_{\psi^2} - \by{R}{2}(\cos\theta\del_x + \sin\theta\del_y - \by{1}{c}\del_{\theta}).
\end{equation}

\subsubsection{Symmetry reduction} 
Since $\A$ is a connection one-form for the principal bundle 
$\pi: Q\toto M = Q/G$, the system $(Q,\mathcal{D},\mathcal{L})$ is of $G$-Chaplygin type.
Let $J: TQ\to\gu^*$ be the \momap of the $G$-action. Then a
calculation shows that 
\[
 \vv<J(q,v^1\xi_1+v^2\xi_2),\curv^{\mathcal{A}}(\del_{\psi^1},\del_{\psi^2})> 
 = m_0l\by{R^3}{4c^2}(v^2-v^1).
\]
Note that this vanishes if $l=0$.  
Let us apply the Gram-Schmidt orthonormalization scheme with respect
to the reduced metric $\mu_0$ to
$\del_{\psi^1},\del_{\psi^2}$ and denote the result by
$u_1,u_2$. 
Thus,
\begin{equation}\label{e:u_i}
 u_1 
 = \left(J_w + m\by{R^2}{4} +
 J_0\by{R^2}{4c^2}\right)^{-\frac{1}{2}}\del_{\psi^1}
\end{equation}
and 
\[
\tag{\ref{e:u_i}a}
 u_2
 = \left(
   J_w\by{\left(J_w+m\by{R^2}{2}+J_0\by{R^2}{2c^2}\right)
   +mJ_0\by{R^4}{4c^2}}{J_w+m\by{R^2}{4}+J_0\by{R^2}{4c^2}}
\right)^{-\frac{1}{2}}
\left(
  \del_{\psi^2} 
   -
   \by{m\by{R^2}{4}-J_0\by{R^2}{4c^2}}{J_w+m\by{R^2}{4}+J_0\by{R^2}{4c^2}}
   \del_{\psi^1}
\right).
\]
Using the relation 
\[
 \mu_0(b)
 = \beta
 = \mu_0\left(\sum_{i,j} \vv< J(u_i),\curv^{\mathcal{A}}(u_j,u_i) > u_j ,
 \_\right)
\]
and expanding everything in terms of $\del_{\psi^1},\del_{\psi^2}$,
one finds that the drift vector $b$ equals 
\begin{equation}\label{e:b-rob}
 b = lm_0R^3\Big(J_w(4c^2J_w + m2c^2R^2 + 2J_0R^2) + mJ_0R^4\Big)^{-1}
      \left(\del_{\psi^1}+\del_{\psi^2}\right).
\end{equation}
Since $M$ is compact, this $b$ cannot be the gradient of a function for
$l\neq0$. Thus, we can conclude that the deterministic two-wheeled robot 
does not have a preserved volume for $l\neq0$ and that the associated 
stochastic system is not time-reversible.

\subsubsection{Kinematics of the noisy cart}
Formula \eqref{e:b-rob} seems to imply that
the stochastic cart (with zero initial velocity) 
acquires a tendency to go backwards
when the center of mass is displaced towards the rear. 
To see that this is indeed the case we should check that the
horizontally lifted mean curve coincides with the expected motion of
the cart.

Since $TQ\cong_{\mu}T^*Q$ (vector bundle isomorphism induced by the
Riemannian metric $\mu$ on $Q$) and $TTQ$ are trivial, we may view 
$TQ\subset TTQ$ as a vector subbundle,
so that $b^h=\hl(b)$ and $u_a^h$ become vector fields on
$TTQ$. Then the stochastic dynamics $\Gamma^{\mathcal{D}}=(q_t,p_t)$ on
$\mathcal{D}$ are generated by the operator
\[
 A^{\mathcal{D}} 
 = X_{\mathcal{H}}^{\mathcal{C}}
   + \by{1}{2}b^h
   + \sum_a u_a^hu_a^h
\]
or by the Stratonovich equation
\[
 \delta\Gamma^{\mathcal{D}}
 = 
 X_{\mathcal{H}}^{\mathcal{C}}(\Gamma^{\mathcal{D}})\delta t
 + \by{1}{2}b^h(q_t)
 + \sum_a u_a^h(q_t)\delta W^a.
\]
Here $X_{\mathcal{H}}^{\mathcal{C}}$ was defined in
Section~\ref{sec:2}. 
Now in local coordinates $(q^i,p_i)$ on $TQ$ 
the stochastic equations of motion are 
\begin{align*}
 \delta(p_i\circ\Gamma^{\mathcal{D}}_t)
 &= dp_i \left(PX_{\mathcal{H}}(q_t,p_t)\right) \delta t, \\
 \delta(q^i\circ\Gamma^{\mathcal{D}}_t)
 &= (p_t)_i\delta t 
  + \by{1}{2}dq^i\left(b^h(q_t)\right)\delta t
  + dq^i\left(\sum_a u^h_a(q_t)\right)\delta W^a.
\end{align*}
If the initial conditions are $(q_0,0)$ then the solution is given by
$(q_t,0)$ where $q_t$ satisfies 
\[
 \delta q_t 
 = \by{1}{2}b^h(q_t)\delta t  +  \sum_a u^h_a(q_t)\delta W^a.
\]
Let $q_t = (\psi^1_t,\psi^2_t,x_t,y_t,\theta_t) = (q^i_t)$. By
\cite[Lemma~7.3.2]{Oks07} we have
\[
 E\left[q^i_t\right] 
 = q_0^t 
 + E\left[\int_0^t\left(\by{1}{2}\left(\sum_a u_a^h u_a^h q^i\right)(q_s) 
     + \by{1}{2}(b^hq^i)(q_s)\right) \,ds\right].
\]

Rewriting 
\[
 \sum_a u_a^h u_a^h
 = A(\xi_1\xi_1+\xi_2\xi_2) + B(\xi_1\xi_2+\xi_2\xi_1),
\]
with
\[
 A := \frac{\mu(\xi_1,\xi_1)}{\mu(\xi_1,\xi_1)^2-\mu(\xi_1,\xi_2)^2}\,,
 \qquad
 B := \frac{\mu(\xi_1,\xi_2)}{\mu(\xi_1,\xi_1)^2-\mu(\xi_1,\xi_2)^2}\,, 
\]
and noting that $\mu(\xi_1,\xi_2)$ and $ \mu(\xi_1,\xi_1) =
\mu(\xi_2,\xi_2)$ are constants, implies that 
$\sum_a u_a^h u_a^h q^i = 0$ for $(q^i) = (\psi^1,\psi^2,x,y,\theta)$.
Thus
\[
 \dd{t}{}E\left[q_t\right] = \by{1}{2}b^h\left(E[q_t]\right), 
\]
i.e., $E[q_t]$ is the
horizontal lift of the integral curve of $\by{1}{2}b\in\X(M)$. 

Therefore, constraints and noise couple to produce a backwards
drift of the robot.
We emphasize that this is a stochastic non-holonomic effect which does
not appear in a Hamiltonian setting. Indeed, the Hamiltonian reduction
of Brownian motion at the zero-momentum level yields Brownian motion
and this is consistent with the fact that the reduced two-wheeled
robot system is actually Hamiltonian when $l=0$.

\subsubsection{Trajectory planning for noisy wheels}
Generally speaking, consider a non-holonomic system (such as the cart)
and assume that it is controlled so as to follow a predefined smooth curve
$c(t)\in Q$, $t\in[0,T]$ when no noise is present. When the system is 
stochastically perturbed we may ask whether $c(t)$ is also the expected 
motion of the perturbed system.

Suppose we want to steer the robot so that it follows a predefined
curve in the plane. As a curve we consider the circle $C$ of radius 
$\rho\ge0$ centered at the origin. The initial configuration of the 
robot should be $(x_0,y_0,\theta_0) = (\rho,0,\pi/2)$ and the vehicle 
should go around the circle in the positive sense. It is assumed that 
the wheel speeds can be individually controlled.

In this section we consider the example of \cite{ZC04}. 
Here the wheels are subject to a Gaussian white noise which is modeled
by the Stratonovich equation
\begin{equation}\label{e:BM-wheels}
 \delta\Gamma^M
 = \sqrt{D_1}\del_{\psi^1}\delta W^1
   + \sqrt{D_2}\del_{\psi^2}\delta W^2
\end{equation}
in $TM$
where $(W^i)$ is Brownian motion in $\mathbb{R}^2$ and $D_i>0$ are constants. 
In this setup one assumes that the controlled vehicle is not affected
by the kinematics of the problem, 
thus effectively forgetting the metric $\mu$. 
The generator of $\Gamma^M$ is
$\by{1}{2}(D_1\del_{\psi^1}^2+D_2\del_{\psi^2}^2)$. 
Equation~\eqref{e:BM-wheels}
lifts to a Stratonovich equation
\[
 \delta\Gamma^Q
 = \sqrt{D_1}\xi_1\delta W^1
   + \sqrt{D_2}\xi_2\delta W^2
\]
in $TQ$. This is in accordance with the general theory of
\cite{ELL04,ELL10}; the generator of $\Gamma^Q$ is
$A^Q = \by{1}{2}(D_1\xi_1^2+D_2\xi_2^2)$ which can be regarded as the
horizontal lift of $A^M$. 
Consider the deterministic input vector field 
\[
 u(t) 
 :=
 -\lam(t)\left(\by{\rho+c}{R}\del_{\psi^1}+\by{\rho-c}{R}
 \del_{\psi^2}\right) 
\]
on $M$ where the control
\[
 \lam(t)
 =
 \left\{
 \begin{matrix}
   2t, && 0\le t <\sqrt{\by{\pi}{2}} =: t_1\\
   2\by{t-T}{t_1-T}, && t_1\le t \le T := \by{3\pi}{2}+t_1
 \end{matrix}
 \right. 
\]
is chosen such that the unperturbed robot traverses the nominal
curve $C$ exactly once and 
initial and final speed are $0$.
The equation for the controlled noisy robot is thus
\[
 \delta\Gamma^u
 =
 - \lam(t)(\by{\rho+c}{R}\xi_1(\Gamma^u)+\by{\rho-c}{R}\xi_2(\Gamma^u))\delta t
 + \sqrt{D_1}\xi_1(\Gamma^u)\delta W^1
 + \sqrt{D_2}\xi_2(\Gamma^u)\delta W^2
\]
and the corresponding (time-dependent) generator is $A^u = A^Q +
\hl(u)$ whence by \cite[Lemma~7.3.2]{Oks07}
\begin{equation}\label{e:A^u}
 E\left[f(\Gamma^u_t)\right] = f(\Gamma_0^u)
 - E\left[\int_0^t(A^Q+\hl(u))(f)(\Gamma^u_s)\,ds\right]
\end{equation}
for $f\in\cinf(Q)$. 
(The expectation is taken with respect to the underlying probability.)
Let 
\[
 \Gamma_0=(0,0,\rho,0,\by{\pi}{2}),
 \quad
 \Gamma_t^u =: (\psi^1_t,\psi^2_t,x_t,y_t,\theta_t),
 \quad
 \kappa := \by{(D_2-D_1)R^2}{8c}.
\]
Using \eqref{e:A^u} we find
\begin{align*}
 E[x_t]
 &= \kappa\int_0^tE[\cos(\theta_s)]\,ds
   + \rho\int_0^t\lam(s)E[\sin(\theta_s)]\,ds\\
 &= \kappa\int_0^te^{\kappa t}\cos(\theta(s))\,ds
   + \rho\int_0^t\lam(s)e^{\kappa t}\sin(\theta(s))\,ds\\
 E[y_t]
 &= - \kappa\int_0^tE[\sin(\theta_s)]\,ds
    + \rho\int_0^t\lam(s)E[\cos(\theta_s)]\,ds\\
 &= - \kappa\int_0^te^{\kappa t}\sin(\theta(s))\,ds
    + \rho\int_0^t\lam(s)e^{\kappa t}\cos(\theta(s))\,ds
\end{align*}
where
$\theta(t)$ differs from $\theta_t$ and is defined by
\[
 \theta(t)
 =   
 \left\{
 \begin{matrix}
 t^2+\by{\pi}{2}, && 0\le t <\sqrt{\by{\pi}{2}} =: t_1;\\
 \by{(t-T)^2}{t_1-T} + \by{5\pi}{2}, && t_1\le t \le T := \by{3\pi}{2}+t_1.
 \end{matrix}
 \right. 
\]
This determines the orientation of the vehicle.

We have solved for $(E[x_t],E[y_t])$ using Maxima and its built in
Runge-Kutta scheme.
Here is a plot:
\begin{center}
 \includegraphics[width=100mm,height=60mm]{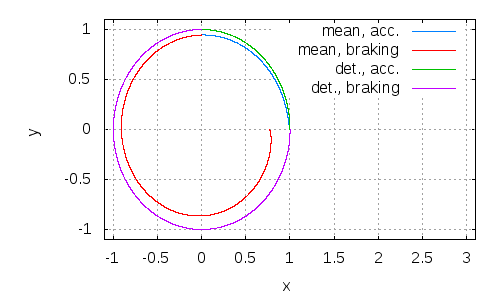}\\
\parbox[center]{12cm}{
The data are $\rho=1$, $D_1=1.2$, $D_2=0.8$, $R=0.3$, $c=0.1$.
we have plotted the accelerating and braking parts of
$(E[x_t],E[y_t])$ as blue and red, and  
the accelerating and braking parts of the unperturbed controlled robot
$(x(t),y(t))$ as green and magenta, respectively.
}
\end{center}
The discrepancy between the deterministic trajectory and the mean
curve of the perturbed system is quite obvious. This phenomenon has
also been observed in \cite{ZC04} by means of numerical simulations,
and \cite{ZC04} have also proposed a trajectory planning algorithm
which takes the perturbation into account.
When comparing the above picture to that of \cite{ZC04}, it should be
noted that we have chosen a different convention for the orientation
of the wheels.

\subsection{Microscopic snakeboard under molecular bombardment}
This is not a $G$-Chaplygin system but does fit the set-up of
Section~\ref{sec:CBM}. 

In describing the snakeboard we follow mostly the presentation of
\cite{CMR01}. There is, however, one difference: 
the metric which we use to define the kinetic energy is that
of \cite{BL02}. This considerably simplifies some of the formulas. 
We further assume that the angle of the front axis equals minus that
of the back axis. Thus 
the configuration space of this system is 
\[
 Q 
 = S^1\times S^1\times\textup{SE}(3)
 = \set{q=(\phi,\psi,x,y,\theta)}.
\]
The constraint distribution is the kernel of the $\mathbb{R}^2$-valued
one-form $\om=(\om_1,\om_2)$ given by
\begin{align*}
 \om_1(q)
 &=
  -\sin(\theta+\phi)dx 
  + \cos(\theta+\phi)dy -
  r\cos(\phi)d\theta,\\
 \om_2(q)
 &=
  -\sin(\theta-\phi)dx 
  + \cos(\theta-\phi)dy -
  r\cos(\phi)d\theta
\end{align*}
where $2r$ is the distance between the axes measured from their
respective midpoints. Thus 
\[
 \mathcal{D}
 = 
 \textup{span}\,\set{
 \del_{\phi},\del_{\psi},s := a\del_x+b\del_y+c\del_{\theta} }
\]
where the functions $a,b,c$ are given by
\begin{align*}
 a
 &=
  -r\big(
  \cos(\phi)\cos(\theta-\phi)+\cos(\phi)\cos(\theta+\phi)
  \big),\\
 b
 &= 
 -r\big(
  \cos(\phi)\sin(\theta-\phi)+\cos(\phi)\sin(\theta+\phi)
  \big),\\
 c
 &=
 \sin(2\phi).
\end{align*}
Let $m$ be the mass of the board, $J_0$ its moment of inertia, and 
$J_{\phi}$, $J_{\psi}$, $J_{\theta}$ the moments of
inertia corresponding to rotation about the angle $\phi$, $\psi$, and
$\theta$ respectively.  Then the Lagrangian
of the system is the kinetic energy of the metric
\[
 \mu
 =
 m(dx\otimes dx + dy\otimes dy)
 + K d\theta\otimes d\theta
 + J_{\phi} d\phi\otimes d\phi
 + J_{\psi} d\psi\otimes d\psi
 + J_{\psi}(d\psi\otimes d\theta + d\theta\otimes d\psi)
\]
where $K := J_{\theta}+J_{\psi}+J_{\phi}$.

Let us assume that the snakeboard is perturbed by white noise.
Using the left trivialization of $TQ$ this can be modeled by a
Stratonovich operator of the form
\begin{align*}
 \mathcal{S}:
 Q\times T\mathbb{R}^6&\longto TQ,\\
 (q,w,w')
 &\longmapsto
 \sigma\sum\vv<e_i,w'>u_i \delta W^i
\end{align*}
where $(u_i)$
is a left invariant orthonormal frame on $Q$ and $\sigma\ge0$ is a
parameter specifying the field strength. According to the results of 
Section~\ref{sec:4}, constrained Brownian motion is a diffusion
$\Gamma^{\textup{nh}}$ with generator
\[
 A = \by{\sigma^2}{2}\sum(u_a u_a - \Pi\nabla^{\mu}_{\Pi u_a}u_a).
\]
Here $(u_a)$ is an orthonormal frame of $\mathcal{D}$ and $\Pi:
TQ=\mathcal{D}\oplus\mathcal{D}^{\bot}\to\mathcal{D}$. We fix this frame to be
\[
 u_1 = J_{\phi}^{-\frac{1}{2}}\del_{\phi},\quad
 u_2  = \eta^{-\frac{1}{2}}(\del_{\psi}-J_{\psi}\by{c}{\eps}s),\quad
 u_3  = \eps^{-\frac{1}{2}}s
\]
where
\[
 \eps = m(a^2+b^2)+Kc^2,
 \quad \eta = J_{\psi}(1-\by{J_{\psi}c^2}{\eps}).
\]
Note that $\eta$ and $\eps$ are functions of $\phi$ only. 
A calculation now shows that we have, for the (trivial) connection
$\nabla$ associated to $\mu$,
\begin{align*}
 \nabla_{u_1}u_1 &= 0,\\
 \nabla_{u_2}u_2 &=
 \by{J_{\psi}^2c^3}{\eta\eps^2}\big((\del_{\theta}a)\del_x+(\del_{\theta}b)\del_y\big)
 \in\mathcal{D}^{\bot},\\
 \nabla_{u_3}u_3
 &=  
 \by{c}{\eps}\big((\del_{\theta}a)\del_x+(\del_{\theta}b)\del_y\big)
 \in\mathcal{D}^{\bot}.
\end{align*}
Thus
$\Pi\nabla^{\mu}_{\Pi u_a}u_a = 0$ for this
frame and $\Gamma^{\textup{nh}}$ is given by the Stratonovich
equation
\begin{equation}\label{e:sb_cbm}
 \delta\Gamma^{\textup{nh}} = \sigma\sum u_a(\Gamma^{\textup{nh}})\delta W^a.
\end{equation}
As in the theory of \cite{CMR01}, we fix the horizontal space of the 
principal bundle $\pi: Q\toto Q/G =
T^2 = M$ associated to the distribution $\mathcal{D}$ to be given by the span of
$\set{u_1,u_2}$.
The corresponding connection form is denoted by $\A$.
Consider the control vector fields 
\begin{align*}
 U_{\phi}(t) &= u_{\phi}'(t)\del_{\phi},\qquad
 u_{\phi}(t) = a_{\phi}\sin(\om_{\phi}t),\\
 U_{\psi}(t) &= u_{\psi}'(t)\del_{\psi},\qquad
 u_{\psi}(t) = a_{\psi}\sin(\om_{\psi}t)
\end{align*}
in the control space $TM$. Their horizontal lifts are $\hl(U_{\phi}) =
u_{\phi}'(t)\del_{\phi}$ and $\hl(U_{\psi}) =
u_{\psi}'(t)(\del_{\psi}-J_{\psi}\by{c}{\eps}s)$. Combining this with
\eqref{e:sb_cbm} yields 
\begin{equation}\label{e:sb_gamma_u} 
 \delta\Gamma^u
 = \hl_{\Gamma^u}(U_{\phi}+U_{\psi})\delta t
   + \sigma u_a(\Gamma^u)\delta W^a
\end{equation}
which describes the stochastic perturbation of the controlled
snakeboard with deterministic gait input 
$(\phi,\psi)=(u_{\phi}(t),u_{\psi}(t))$.

Since the variables $(\phi,\psi)$ are also the ones
which can be controlled, we are interested in estimating $\Gamma^u$ given
that the projected process $X_t = \pi\circ\Gamma^u_t$ satisfies the
projected equation
\begin{equation}\label{e:sb_X}
 \delta X = (U_{\phi}(t)+U_{\psi}(t))\delta t
            + J_{\psi}^{-\frac{1}{2}}\delta W^1 \del_{\phi} 
            + \eta(X)^{-\frac{1}{2}}\delta W^2 \del_{\psi}\,.
\end{equation}
This is the filtering problem
$E[\Gamma^u_t|\pi\circ\Gamma^u_t = X_t] =: Z_t$ and the solution 
is provided by \cite{ELL04,ELL10}: The process $\Gamma^u$ can be
decomposed as 
\[
 \Gamma^u_t = g^{X}_t\cdot X_t^h
\]
where $X_t^h$ is the horizontal lift of $X_t = \pi\circ\Gamma^u_t$ and 
$g^X$ is the
reconstruction process. 
These satisfy the Stratonovich
equations
\[
 \delta X^h 
   = \hl_{X^u}(U_{\phi}+U_{\psi})\delta t
   + \sigma(u_1(X^h)\delta W^1 + u_2(X^h)\delta W^2),
   \qquad X^h_0 = \Gamma^u_0 = q_0\in Q
\]
and
\[
 \delta g^{X}_t 
 = \sigma T_eL_{g^X_t}.\A_{X_t^h}u_3(X^h_t)\delta W^3,
 \qquad
 g^X_0 = e \in G.
\]
(See also Section~\ref{sec:equiv-diff}.)
By \cite{ELL10} we have that 
\begin{equation}\label{e:Z}
 Z_t = E[g_t^X]\cdot X_t^h.
\end{equation}
Let $X_t^h = (\phi_t,\psi_t,x_t,y_t,\theta_t)$ and $E[g_t^X] =
(a_t,b_t,\gamma_t)\in G$. 
It follows from Proposition~\ref{prop:mean-re} that the mean reconstruction curve
$E[g_t^X]$ is determined by the time- and $\om\in\Om$-dependent ODE 
\begin{align}\label{e:E[g]}
 \dd{t}{}E[g_t^X]
 &= 
 \dd{t}{}
 \left(
  \begin{matrix}
    a_t\\
    b_t\\
    \gamma_t
  \end{matrix}
 \right)\\
 &=
 \by{\sigma^2c(\phi_t)}{2\sqrt{\eps(\phi_t)}}
 \left(
  \begin{matrix}
    -\big(a(\phi_t,\theta_t)+y_tc(\phi_t)\big)\sin(\gamma_t) 
      -\big(b(\phi_t,\theta_t)-x_tc(\phi_t)\big)\cos(\gamma_t)\\
    \big(a(\phi_t,\theta_t)+yc(\phi_t)\big)\cos(\gamma_t)
      -\big(b(\phi_t,\psi_t)-x_tc(\phi_t)\big)\sin(\gamma_t)\\
    0
  \end{matrix}
 \right). \notag
\end{align}
Using the rule for transforming Stratonovich equations to It\^o type, we
can characterize $X_t^h$ by the It\^o equation
\begin{multline}\label{e:XIto}
 dX_t^h
 =
 \left(
 \begin{matrix}
 \vspace{1mm}
   u_{\phi}'(t)\\
   \vspace{1mm}
   u_{\psi}'(t)\\
   \vspace{1mm}
   -J_{\psi}\by{a(\phi_t,\theta_t)c(\phi_t)}{\eps(\phi_t)}u_{\psi}'(t)
     + \sigma^2\by{c(\phi_t)^3}{2\eta(\phi_t)\eps(\phi_t)^2}(\del_{\theta}a)(\phi_t,\theta_t)\\ 
     \vspace{1mm}
   -J_{\psi}\by{b(\phi_t,\theta_t)c(\phi_t)}{\eps(\phi_t)}u_{\psi}'(t)
     + \sigma^2\by{c(\phi_t)^3}{2\eta(\phi_t)\eps(\phi_t)^2}(\del_{\theta}b)(\phi_t,\theta_t)\\ 
   -J_{\psi}\by{c(\phi_t)^2}{\eps(\phi_t)}u_{\psi}'(t)
 \end{matrix}
 \right)dt\\
 + \sigma u_1(X^h)dW^1 + \sigma u_2(X^h)dW^2.
\end{multline}
Equation~\eqref{e:XIto} involves an iterated dependence on 
trigonometric functions, and hence numerical simulation is not straightforward.
A naive approach would involve to run an Euler-Maruyama and an Euler simulation for
\eqref{e:XIto} and \eqref{e:E[g]} respectively, and to multiply the
results together according to \eqref{e:Z} which is the action of $G$
on $Q$. This yields $Z_t$. Running the simulation sufficiently many  
times and computing the average yields the mean $E[Z_t]$. We have 
implemented this scheme and the results seem reasonably stable up to 
time 1, according to a first order test. Beyond that time,
the trajectories blow up very quickly, which is a strong indication
that the method is unstable and a more detailed analysis of the
numerical implementation is necessary. Our preliminary results are contained in the plot below.
\begin{center}
 \includegraphics[width=11cm,height=9cm]{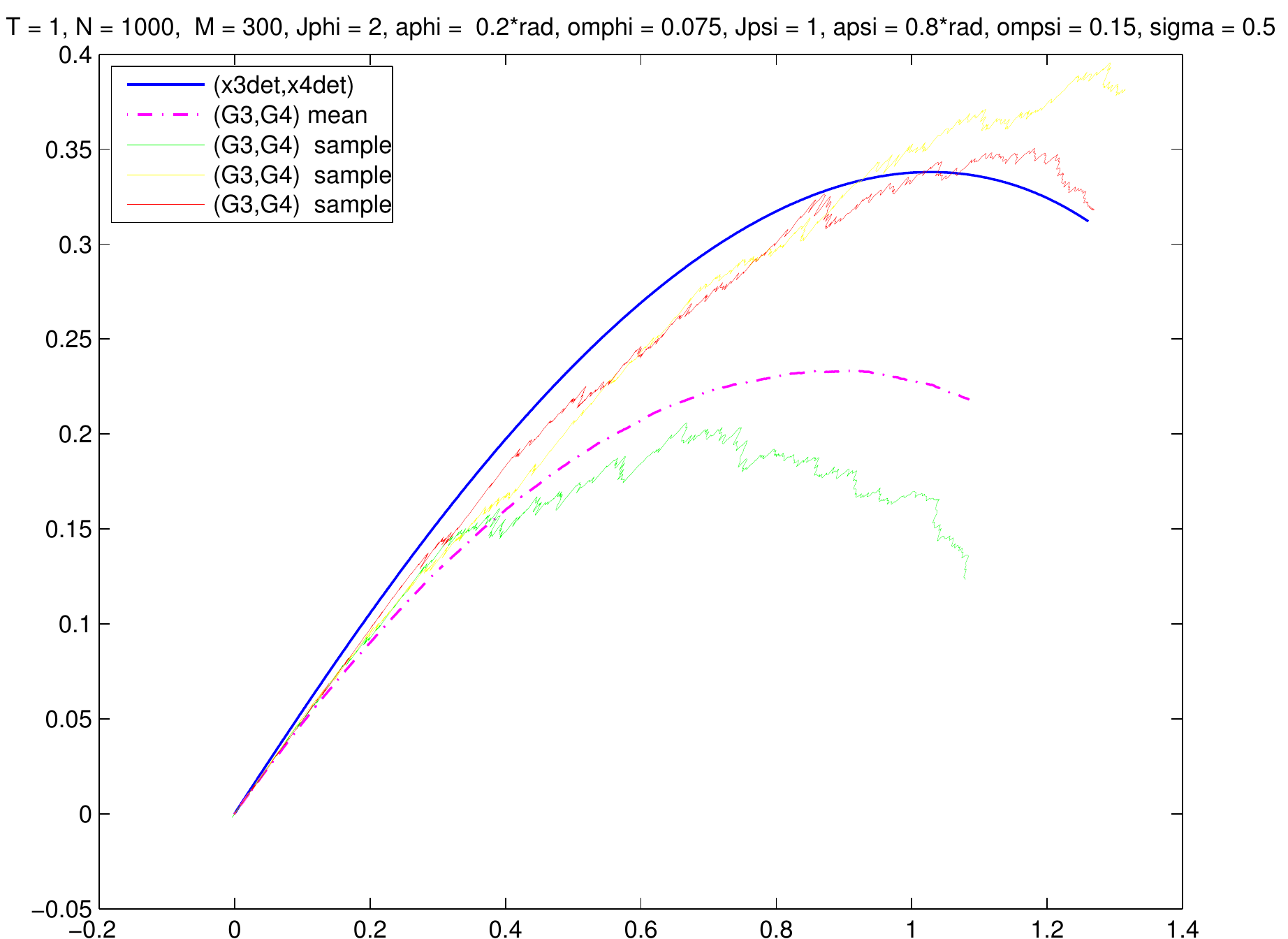}
 \\
\parbox[center]{12cm}{
The blue line is the center of mass motion of the unperturbed
snakeboard and and the dotted magenta line is the mean motion of the
stochastic snakeboard with the same deterministic input.
Additionally $3$ sample plots have been included.
The data are as indicated above: $T$ is the runtime, $1/N$ the
step size, $M$ the number of experiments, $\textup{rad} = \by{180}{\pi}$ 
and $\sigma$ the parameter
specifying the strength of the white noise.
The initial conditions are $q_0=(0,0,0,0,0.5)$. 
}
\end{center}

\noindent\textit{Acknowledgments.}
T.S. Ratiu was partially supported by Swiss NSF grant 200021-140238 
and by the government grant of the Russian Federation for support of 
research projects implemented by leading scientists, Lomonosov Moscow 
State University under the agreement No. 11.G34.31.0054.

\end{document}